\documentclass[aps,pra,twocolumn,footinbib,
showpacs,showkeys,longbilbliography]{revtex4-1}

\usepackage{graphicx}
\usepackage{indentfirst}
\usepackage{braket}
\usepackage{float}
\usepackage{amsmath}
\usepackage{amssymb}
\usepackage{verbatim}
\usepackage{wasysym}
\usepackage{epstopdf}
\usepackage{CJK}
\usepackage{esint}
\usepackage{color}
\usepackage{xcolor}
\usepackage{epsfig}
\usepackage{subfigure}
\usepackage{amsfonts}
\usepackage{footmisc}
\usepackage{scrextend}
\usepackage{multirow}
\usepackage{mathtools}

\usepackage[h]{esvect}

\usepackage{xr-hyper}
\usepackage[hyperfootnotes=false]{hyperref}

\usepackage[english]{babel}
\usepackage{url}
\usepackage{bm}
\definecolor{darkblue}{rgb}{0,0,0.5}
\hypersetup{
colorlinks=true,
linkcolor=black,
filecolor=black,
citecolor=darkblue,
urlcolor=darkblue,
}

\newcommand{\defeq}{\vcentcolon=}

\newcommand\mc[1]{\mathcal{#1}}

\newcommand\bs[1]{\boldsymbol{#1}}
\newcommand\bff[1]{\text{\textbf{#1}}}

\DeclareMathOperator*{\argmax}{arg\,max}
\DeclareMathOperator*{\argmin}{arg\,min}

\newtheorem{theorem}{Theorem}

\newtheorem{corollary}{Corollary}

\newenvironment{proof}[1][Proof]{\noindent\textbf{#1.} }{\ \rule{0.5em}{0.5em}}

\begin{document}

\title{End-to-End Capacities of Imperfect-Repeater Quantum Networks}
\author{Cillian Harney}
\email{cth528@york.ac.uk}
\author{Stefano Pirandola}
\email{stefano.pirandola@york.ac.uk}
\affiliation{Department of Computer Science, University of York, York YO10 5GH, United Kingdom}
\begin{abstract}
The optimal performance of a communication network is limited not only by the quality of point-to-point channels, but by the efficacy of its constituent technologies. Understanding the limits of quantum networks requires an understanding of both the ultimate capacities of quantum channels and the efficiency of imperfect quantum repeaters. 
In this work, using a recently developed node-splitting technique which introduces internal losses and noise into repeater devices, we present achievable end-to-end rates for noisy-repeater quantum networks. These are obtained by extending the coherent and reverse coherent information (single channel capacity lower bounds) into end-to-end capacity lower bounds, both in the context of single-path and multi-path routing. These achievable rates are completely general, and apply to networks composed of arbitrary channels arranged in general topologies. Through this general formalism, we show how tight upper-bounds can also be derived by supplementing appropriate single-edge capacity bounds. As a result, we develop tools which provide tight performance bounds for quantum networks constituent of channels whose capacities are not exactly known, and reveal critical network properties which are necessary for high-rate quantum communications.
This permits the investigation of pertinent classes of quantum networks with realistic technologies; qubit amplitude damping networks and bosonic thermal-loss networks.
\end{abstract}
\maketitle

\section{Introduction}

The theoretical and experimental investigation of quantum networks \cite{SlepianNets, CoverThomas, TanenbaumNets, GamalNets, KimbleQI,UniteQInt, RazaviQNet} is essential for the deployment of quantum information technologies on a global scale. 
Quantum networks will not only enable provably secure communication through quantum key distribution (QKD) \cite{GisinQCrypt, AdvCrypt}, but will facilitate distributed quantum information processing. As such, a global quantum internet will be pivotal to the expansion and empowerment of future quantum technologies.

Yet, the ability to transmit, detect and preserve quantum information is remarkably more difficult than that of classical information. Classical information is robust, can be copied and stored reliably. On the other hand, quantum information is inherently fragile and the laws of quantum mechanics prohibit its cloning \cite{Mike_Ike}. Hence, quantum networks face many formidable obstacles from which their classical counterparts are spared. These can be attributed to two key regimes: Inevitable \textit{external} decoherence experienced by quantum systems along communication channels; and \textit{internal} decoherence due to imperfect devices attempting to preserve or operate on quantum systems. 

The impact of external decoherence on the performance of quantum communication networks is understood via channel capacities. For a quantum channel $\mc{E}$, the capacity $\mc{C}(\mc{E})$ describes the absolute maximum rate that entanglement, secret-keys or quantum states can be distributed using perfect transmitters/receivers by either user. Significant advancements in study of channel capacities have propelled our understanding of the limits of quantum communications, notably in optical-fibre \cite{PirPatron09, PLOB} and free-space \cite{FS,SQC}. Fundamental limitations have also been unveiled for quantum networks composed of teleportation-covariant channels \cite{End2End, MultiEnd}; providing invaluable tools for the investigation of quantum network capacities for both random \cite{QuntaoRandQNets, ZhangQInt} and ideal \cite{OPGQN} optical-fibre networks. 

Previous investigations have operated under the assumption of ideal repeaters. Unfortunately, there will always exist unavoidable internal decoherence due to sub-optimal detection or transmission, imperfect quantum memories, electronic and environmental noise, and more. The nature of this decoherence depends on the type of quantum system being used, whether one is using discrete-variable \cite{Mike_Ike, WatrousTxt, Holevo19} (DV) or continuous-variable (CV) \cite{RalphCV, BraunsteinVL, GaussRev,SerafiniCV} quantum information and the protocol being employed. Nonetheless, such effects must be considered to identify ultimate limits for realistic technologies. 

Recently, the performance limits of quantum networks with imperfect repeaters have been studied using a \textit{node-splitting} technique; each channel in the network is split into compound channels which incorporate repeater imperfections via additional internal channels \cite{RicSplitting}. This work focussed on the impact of internal loss and the class of distillable channels \cite{PLOB}, but it is possible to extend these results into a more general setting. In this work, we provide a framework for bounding the end-to-end capacities of noisy-repeater quantum networks. We begin by translating the coherent information (CI) and reverse coherent information (RCI) from point-to-point achievable quantum communication rates \cite{PirPatron09} into lower-bounds on end-to-end network capacities. The resulting bounds are universal, regardless of network topology or channel composition. Combined with relevant upper-bounds, we are able to apply a node-splitting technique which accounts for both internal loss and noise. In particular, we unveil realistic performance limits and infrastructure requirements of quantum networks composed of channels whose capacities are not exactly known. 

This paper proceeds as follows: Section~\ref{sec:Prelims} details some preliminary notions, covering the basics of quantum networks, single/multi-path end-to-end capacities and useful architectures. In Section~\ref{sec:GenBounds} we present general end-to-end capacity bounds for quantum networks with arbitrary topologies. In particular, we present achievable end-to-end rates which are universal for any quantum network, and detail how upper-bounds can be found when single-edge capacity bounds are known. Furthermore, we describe the node-splitting technique and discuss how these tools can be used to capture critical properties of quantum network architectures. These methods and results are then applied in the context of qubit amplitude-damping and bosonic thermal-loss weakly-regular networks; a class of highly connected and analytically treatable architectures. Finally, Section~\ref{sec:Conc} provides concluding remarks and future directions of study. 

\section{Preliminaries\label{sec:Prelims}}
\subsection{Quantum Networks, Physical and Logical Flow}
A quantum network can be described using an underlying, undirected graph $\mc{N} = (P,E)$ where $P$ is a set of network nodes, and $E$ is a set of edges that inter-connect these nodes throughout the network. A node $\bs{x} \in P$ contains a local register of quantum systems, and may represent a repeater station or a user occupied station.  A network edge exists if the unordered pair $(\bs{x},\bs{y})$ is contained within the total edge set $E$. If $(\bs{x},\bs{y})\in E$ then a quantum communication channel $\mc{E}_{\bs{xy}}$ exists between $\bs{x}$ and $\bs{y}$ through which quantum systems can be exchanged. 

There is important nuance surrounding the reason why quantum networks are treated as undirected graphs. Under the assistance of two-way classical communications (CCs), the optimal transmission of quantum information is connected with optimal entanglement distribution. It does not depend on the direction of physical system exchange (physical flow) but on the local operations (LOs) that are applied at each point, and thus the direction of teleportation. This refers to the \textit{logical direction} of quantum communication (or logical flow). The physical orientation of any channel $\mc{E}_{\bs{xy}}$ in the network can be forwards or backwards and still facilitate communication in either logical direction. 

If a channel is physically symmetric then the physical forward and backward channels between two network nodes $\bs{x},\bs{y} \in P$ are identical $\mc{E}_{\bs{xy}} = {\mc{E}_{\bs{x}\rightarrow\bs{y}} = \mc{E}_{\bs{y}\rightarrow\bs{x}}}$. Networks built up of physically symmetric channels are easy to represent on an undirected graph as the physical channel direction is then completely irrelevant.
But we must take particular care when there exist channels that are \textit{physically asymmetric}, i.e.~the physical forward and backward channels are not identical, ${\mc{E}_{\bs{x}\rightarrow\bs{y}} \neq \mc{E}_{\bs{y}\rightarrow\bs{x}}}$. In this case, a quantum network can be described as a collection of forward and backward channels associated with each undirected edge, $\{ (\mc{E}_{\bs{x}\rightarrow\bs{y}}, \mc{E}_{\bs{y}\rightarrow\bs{x}})\}_{(\bs{x},\bs{y}) \in E}$. 

One may assume that it is sensible to convert the undirected network graph into a directed graph to account for the unique physically directed channels. Yet, the logical flow of information remains independent from physical flow of quantum systems, as it is always possible to invoke a teleportation protocol in the opposite direction to physical quantum system exchange. But if a channel is physically asymmetric there will be an imbalance between the capacity of the forward and backward channels, since they are generally unique. As a result, there will exist an optimal physical direction in which quantum systems should be exchanged.

This identifies a crucial optimisation step of the physical orientation of a quantum network, and corresponds to fixing an optimal physical direction to all channels within the network. 
Consider an arbitrary edge $(\bs{x},\bs{y})\in E$ and its associated pair of physical forward and backward channels $ (\mc{E}_{\bs{x}\rightarrow\bs{y}}, \mc{E}_{\bs{y}\rightarrow\bs{x}})$. The directed channel pair should be mapped to the channel which maximises the point-to-point capacity,
\begin{align}
 (\mc{E}_{\bs{x}\rightarrow\bs{y}}, \mc{E}_{\bs{y}\rightarrow\bs{x}}) \mapsto \mc{E}_{\bs{xy}}^* \defeq \hspace{-3mm}\underset{\mc{E}\in\{\mc{E}_{\bs{x}\rightarrow\bs{y}}, \mc{E}_{\bs{y}\rightarrow\bs{x}}\}}{\argmax} {\mc{C}}(\mc{E}).
 \label{eq:CapMax}
\end{align}
where $\mc{C}(\mc{E})$ is the capacity of the channel $\mc{E}$. By performing this optimisation for all $(\bs{x},\bs{y})\in E$, the network can be represented as an undirected graph $\mc{N} = (P,E)$ interconnected by the optimal physically directed set of channels $\{ \mc{E}_{\bs{xy}}^{*} \}_{(\bs{x},\bs{y})\in E}$ for quantum communication throughout the network. 

\begin{figure*}
\includegraphics[width=0.95\linewidth]{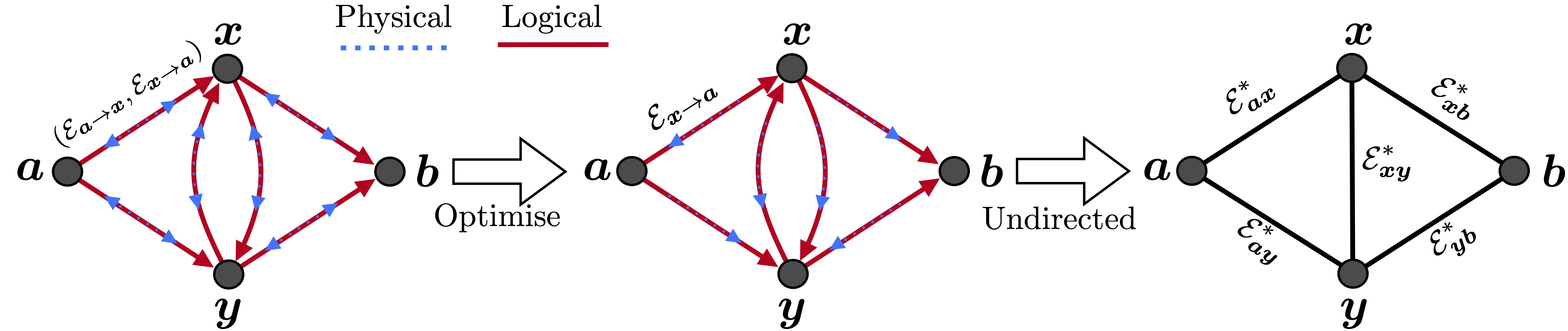}
\caption{Optimising the physical orientation of a quantum network. The physical flow of a quantum systems is independent of the logical flow of information, and thus physical directions can always be optimised to utilise the directed channels which possess the greatest capacity. In this way, a quantum network described by a directed graph can always be reduced to an undirected form. From a pair of physically-directed channels on a single edge $(\mc{E}_{\bs{a}\rightarrow\bs{x}}, \mc{E}_{\bs{x}\rightarrow\bs{a}})$, we can choose that which as the greatest capacity according to Eq.~(\ref{eq:CapMax}), e.g.~$\mc{E}_{\bs{x}\rightarrow{\bs{a}}}$ in the example above. It is always possible to use this superior edge for logical communication in either direction, reducing it to an optimal undirected edge describing the channel $\mc{E}_{\bs{ax}}^* = \mc{E}_{\bs{x}\rightarrow{\bs{a}}}$.}
\label{fig:explain_transform}
\end{figure*}

\subsection{Network Cuts}

Consider a general quantum network $\mc{N}=(P,E)$ which has been pre-optimised to adopt some optimal physical orientation of quantum channels $\{\mc{E}_{\bs{xy}}^*\}_{(\bs{x},\bs{y})\in E}$. Let denote an end-user pair who wish to communicate using the two element set of nodes $\bs{i} = \{ \bs{\alpha}, \bs{\beta} \}$, such that $\bs{\alpha} \in P$ is Alice and $\bs{\beta}\in P$ is Bob. Collecting the end-user pair into a single object $\bs{i}$ is useful for subsequent notation. The goal of end-to-end communication is then to utilise network resources to distribute secret-keys, entanglement or quantum states between $\bs{\alpha}$ and $\bs{\beta}$; all other nodes on the network can be treated as repeaters serving this purpose. 

A vital tool for the study of end-to-end performance is the concept of \textit{network cuts}. We define a network cut $C$ as a means of partitioning a quantum network into two, disjoint super-users: Super-Alice $\bff{A}$ and super-Bob $\bff{B}$ such that $P = \bff{A} \cup \bff{B}$ and $\bs{\alpha}\in\bff{A}$, $\bs{\beta}\in\bff{B}$. The cut $C$ generates an associated cut-set $\tilde{C}$ which collects all of the network edges that must be removed in order to consolidate the partition. The cut-set satisfies
\begin{equation}
{\tilde{C} = \{ (\bs{x},\bs{y})\in E ~|~\bs{x}\in \bff{A}, \bs{y} \in \bff{B}\}} 
\end{equation}
Consequently, the cut partitions the network node and edge sets $(P,E)$ in the following way
\begin{equation}
P \xrightarrow{\text{Cut: } C} \bff{A} \cup \bff{B},~~~E \xrightarrow{\text{Cut: } C} E \setminus \tilde{C}.
\end{equation}
Network cuts are vital to many network optimisation tasks and are used to derive the ultimate limits of network routing protocols.

There are two main routing strategies that facilitate end-to-end communication over quantum networks and which find analogy with classical network theory: Single-path and multi-path routing \cite{GamalNets,End2End}.
Single-path routing describes a network protocol in which quantum systems are exchanged from node-to-node throughout the network one after another. This sequential strategy forges a unique path (or end-to-end route) of interactions through the network, and continues until quantum communication has been established between the end-users. A repeater-chain is an example of a network which utilises single-path routing. The single-path network capacity ${\mc{C}}^s(\bs{i}, \mc{N})$ defines the optimal rate at which communication can be achieved through this routing method. This is computed by determining the minimum network cut $C_{\min}$ that generates the smallest, maximum single-edge capacity in the cut set. We may compute the single-path capacity as \cite{End2End}
\begin{equation}
{\mc{C}}^{s}(\bs{i}, {\mc{N}}) = \min_C \max_{(\bs{x},\bs{y})\in \tilde{C}}  \mc{C}(\mc{E}_{\bs{xy}}^*).   \label{eq:SPR}
\end{equation}
Finding the optimal route is equivalent to solving the well known widest-path problem \cite{Dijkstra}.

One can instead design a protocol that uses multi-point network interactions. Rather than operating sequentially, network nodes may exchange multiple quantum systems with many receiver nodes in an effort to explore many more possible routes from $\bs{\alpha}$ to $\bs{\beta}$ and vice versa. Multi-path routing enables end-users to use multiple end-to-end routes simultaneously, improving both the robustness and rate capabilities of communication. The optimal multi-path strategy is called a \textit{flooding protocol}, in which all channels are used precisely once per end-to-end transmission. This is achieved by performing non-overlapping point-to-multi-point transmissions at each network node, such that receiving nodes only choose to transmit along unused edges for subsequent connections. The multi-path network capacity (or flooding capacity) ${\mc{C}}^m(\bs{i}, \mc{N})$ captures the optimal end-to-end performance out of all possible strategies.

The flooding capacity is found by locating a network cut which minimises a multi-edge capacity over all possible cut-sets. Here, we define a multi-edge capacity of a cut $C$ as the sum of all the single-edge capacities in the corresponding cut-set, 
\begin{equation}
\mc{C}^m(C) \defeq \sum_{(\bs{x},\bs{y})\in \tilde{C}} \mc{C}(\mc{E}_{\bs{xy}}^*).
\end{equation}
The flooding capacity is therefore found by locating the minimum network cut $C_{\min}$ which minimises this quantity \cite{End2End},
\begin{equation}
{\mc{C}}^{m}(\bs{i}, \mc{N}) = \min_C \mc{C}^m(C) = \min_C \sum_{(\bs{x},\bs{y})\in \tilde{C}} \mc{C} (\mc{E}_{{\bs{xy}}}^*) \label{eq:MPR}.
\end{equation}
Computing the flooding capacity is equivalent to solving the max-flow min-cut problem. For a general network architecture there are many efficient numerical techniques \cite{FordFlow,KarpFlow,OrlinFlow}. 

However, it is always possible to place an upper-bound on the flooding capacity by considering an intuitive network-cut. Let the node (edge) neighbourhood of a network node $\bs{x} \in P$ be define as the collection of nodes (edges) that are directly connected to $\bs{x}$. We denote the node and edge neighbourhood sets respectively via,
\begin{align}
&N_{\bs{x}} \defeq \{ \bs{y} ~|~\forall (\bs{x},\bs{y})\in E \},\\
&E_{\bs{x}} \defeq \{ (\bs{x},\bs{y}) ~|~\forall \bs{y} \in N_{\bs{x}} \}.
\end{align}
Given an end-user pair $\bs{i} = \{ \bs{\alpha}, \bs{\beta}\}$ it is always possible to perform a cut that partitions $\bs{\alpha}$ and $\bs{\beta}$ by removing all of the edges in either of their user neighbourhoods, $E_{\bs{\alpha}}$ or  $E_{\bs{\beta}}$. Every end-to-end route must start and finish with user connected edges, therefore the removal of $E_{\bs{\alpha}}$ or $E_{\bs{\beta}}$ eliminates the ability to complete an end-to-end route. We call this kind of network-cut \textit{user node isolation}. User-node isolation can always be used to place (at least) an upper-bound on the end-to-end flooding capacity,
\begin{equation}
{\mc{C}}^{m}(\bs{i}, \mc{N}) \leq \mc{C}_{\mc{N}_{\bs{i}}}^m \defeq \min_{\bs{j}\in\bs{i}} \sum_{(\bs{x},\bs{y})\in E_{\bs{j}}} \mc{C} (\mc{E}_{{\bs{xy}}}^*) \label{eq:FloodUB},
\end{equation}
where we have defined $\mc{C}_{\mc{N}_{\bs{i}}}^m$ as the \textit{min-neighbourhood capacity}, i.e.~the minimum capacity generated by cutting either user neighbourhood \cite{OPGQN}. Since user-node isolation is a valid network-cut the upper-bound is always achievable.

\subsection{Weakly-Regular Quantum Networks \label{sec:WRNs}}

A challenging aspect of benchmarking end-to-end quantum communications via Eqs.~(\ref{eq:SPR}) and (\ref{eq:MPR}) involves designing appropriate network architectures on which the single-path and multi-path capacities can be computed efficiently. Recently proposed in Ref.~\cite{OPGQN}, we make use of \textit{weakly-regular} (WR) architectures. Weakly-regular networks (WRNs) are capable of modelling large-scale, highly-connected quantum networks with a good degree of spatial and topological freedom; all the while permitting analytical derivations of optimal end-to-end rates and critical network properties that contribute to high performance. 

Consider an undirected graph $\mc{N} = (P,E)$ representing a network architecture. Expectedly WRNs admit the graph theoretic property of \text{regularity}; every node $\bs{x}\in P$ has the same degree $\text{deg}(\bs{x}) = k$ for all $\bs{x}\in P$. In other words, the nodal neighbourhood of every network node has the same number of elements,
\begin{equation}
\text{deg}(\bs{x}) = | N_{\bs{x}} | = k, \text{ for all } \bs{x} \in P. 
\end{equation}
Regularity offers an important simplification for analytical investigation. Indeed, general degree distributions are often a key source of complex behaviour in network models. 

\textit{Weakness} of regularity is a more intricate feature. While the degree of a node measures the number of other nodes that it is connected to, a secondary measure that helps to capture how connections are shared between pairs of nodes is called the \textit{commonality}. The commonality of any pair of nodes $\bs{x}, \bs{y}$ counts the number of nodes that both $\bs{x}$ and $\bs{y}$ are commonly connected to. That is, the commonality between any pair of nodes can be computed as the cardinality of the overlap of their node neighbourhood sets,
\begin{equation}
\text{com}({\bs{x}}, {\bs{y}}) = | N_{\bs{x}} \cap N_{\bs{y}} |.
\end{equation}
One can imagine that for a large-scale network where nodes may be very distant, large proportion of nodes will possess a commonality of zero, i.e.~they will not share any neighbours. Hence, it is useful to define a slightly more specific quantity called the \textit{adjacent-commonality}, which counts the number of common neighbours shared by a pair of directly connected nodes (those which share an edge). We denote the adjacent-commonality between the adjacent nodes $\bs{x},\bs{y}$ via $\lambda_{\bs{x}}^{\bs{y}}$ such that
\begin{equation}
\lambda_{\bs{x}}^{\bs{y}} = \lambda_{\bs{y}}^{\bs{x}} = | N_{\bs{x}} \cap N_{\bs{y}} | \iff (\bs{x},\bs{y})\in E.
\end{equation}
In a strongly-regular graph, $\lambda_{\bs{x}}^{\bs{y}}$ is constant for all pairs of network nodes, resulting in a very restricted and small architecture. In contrast, weak-regularity permits a looser characterisation of commonality parameters. Hence, strength or weakness of regularity refers to the consistency of the neighbour sharing properties between pairs of nodes throughout the network. 

In this work, we consider weak-regularity in the following way: Any node $\bs{x} \in P$ in a $k$-WR graph may have a unique \textit{adjacent commonality multi-set} (a modified, potentially degenerate set) which counts the number of neighbours shared between $\bs{x}$ and all its neighbours $\bs{y} \in N_{\bs{x}}$. More precisely, we define the multi-set for any node $\bs{x}$,
\begin{equation}
\bs{\lambda}_{\bs{x}} \defeq \{ \lambda_{\bs{x}}^{\bs{y}_1}, \lambda_{\bs{x}}^{\bs{y}_1}, \ldots \lambda_{\bs{x}}^{\bs{y}_k}\} = \{ \lambda_{\bs{x}}^{\bs{y}} \}_{\bs{y}\in N_{\bs{x}}},
\end{equation} 
Any node in any graph (WR or not) possesses an adjacent commonality multi-set. Here we choose to focus on graphs for which each $\bs{\lambda}_{\bs{x}}$ is contained within a known spectrum of multi-sets. That is, there exists a non-degenerate super-set of permitted adjacent commonality multi-sets,
\begin{equation}
\bs{\Lambda} \defeq \{ \bs{\lambda}_{1}, \bs{\lambda}_{2}, \ldots, \bs{\lambda}_{M} \},
\end{equation} 
so that the adjacent commonality multi-set of any node in the network $\bs{x}\in P$ belongs to the superset $\bs{\lambda}_{\bs{x}} \in \bs{\Lambda}$. One can derive $\bs{\Lambda}$ for any graph, but the consistency of regularity simplifies the number of possible multi-sets.

Consequently, we refer to $(k,\bs{\Lambda})$-WRNs as a class of network for which all nodes have a constant degree equal to $k$ and for which their neighbour-sharing properties satisfy $\bs{\lambda}_{\bs{x}} \in \bs{\Lambda}$, for all $\bs{x}\in P$. While $k$ and $\bs{\Lambda}$ impose connectivity constraints, WRNs that belong to this class are free to adopt a vast range of topological or spatial configurations. Furthermore, we avoid explicit referencing of the number of network nodes $n$. Instead, $n$ is encoded into properties of the network such as the nodal density $\rho_{\mc{N}}$ which defines the average number of network nodes per unit of area (studied in subsequent sections).

Using WRNs we are able to analytically describe and construct useful network structures. For example, it is easy to design a WR \textit{network cell}; a small subgraph which can be iteratively concatenated to generate a large-scale network obeying weak-regularity within some nodal boundary. This permits the analytical investigation of networks consisting many nodes which are highly-connected and display realistic properties. For more details on the complete characterisation of WRNs, we refer the reader to the Supplementary Material of Ref.~\cite{OPGQN}.

\section{Bounds for Realistic Quantum Networks \label{sec:GenBounds}}

Computing end-to-end capacities in Eqs.~(\ref{eq:SPR}) and (\ref{eq:MPR}) requires knowledge of the exact single-edge capacities of all channels in the network, $\mc{C}(\mc{E}_{\bs{xy}}^*)$. For some important classes of quantum channels, exact capacities are known; such as the class of distillable quantum channels, which include the bosonic lossy channel, quantum-limited amplifiers and dephasing channel \cite{PLOB}. As a consequence, their end-to-end network capacities have been fully characterised for arbitrary topologies \cite{End2End}. In general, this is not the case and the capacities of many quantum channels remain undetermined. Hence, we rely on tight upper and lower-bounds in order to understand their efficacy for quantum communication. In this Section we extend point-to-point channel capacity bounds into end-to-end capacity bounds to characterise the performance of general quantum networks.

\subsection{Achievable Rates}
It is always possible to express lower-bounds on a quantum channel capacity using the coherent information (CI) and the reverse coherent information (RCI) \cite{PirPatron09}. Consider a quantum channel $\mc{E}$ and a maximally entangled bipartite EPR state $\Phi_{AB}$, where $A$ and $B$ denote the constituent quantum systems. The Choi matrix of $\mc{E}$ is the result of passing the sub-systems $B$ through the quantum channel, while keeping $A$ preserved i.e.~$\rho_{\mc{E}} \defeq \mc{I}\otimes \mc{E}(\Phi_{AB})$, where $\mc{I}$ denotes the identity channel. 
The CI and RCI of the channel read \cite{PirPatron09}
\begin{align}
I_{\text{C}}(\mc{E}) &\defeq S\left[ \text{Tr}_{A}(\rho_{\mc{E}}) \right] - S(\rho_{\mc{E}}),\\
I_{\text{RC}}(\mc{E}) &\defeq S\left[ \text{Tr}_{B}(\rho_{\mc{E}}) \right] - S(\rho_{\mc{E}}).
\end{align}
where $S(\cdot)$ is the von Neumann entropy. 
These are valuable quantum information theoretic quantities which suggest achievable rates for forward (CI) and backward (RCI) one-way entanglement distillation. Indeed, for any quantum channel $\mc{E}$, the generic two-way assisted capacity $\mc{C}(\mc{E})$ can always be lower-bounded by the hashing inequality \cite{DevWinter},
\begin{equation}
\mc{C}(\mc{E}) \geq {I}(\mc{E}) \defeq \max\left\{ I_{\text{C}}(\mc{E}), I_{\text{RC}}(\mc{E}) \right\}, \label{eq:HashIneq}
\end{equation}
where we have introduced the quantity $I(\mc{E})$ which implicitly maximises over the CI and RCI. This lower-bound is true for any quantum channel \footnote{
For bosonic quantum systems the Choi matrix is energy unbounded, since the maximally entangled state is an infinitely squeezed Two-Mode Squeezed Vacuum (TMSV) state. Thus, additional care must be taken when computing these quantities. Let us denote an energy constrained TMSV state $\Phi^{\mu}$, where $\mu \defeq \bar{n}+1/2$ and $\bar{n}$ is the mean photon number per mode. Then the maximally entangled state takes the form $\Phi_{AB} \defeq \lim_{\mu\rightarrow \infty} \Phi_{\mu}$. 
As a result, we can define an asymptotic Choi matrix as the sequence of finite-energy Choi approximation in the limit of infinite squeezing, $
\rho_{\mc{E}} = \lim_{\mu\rightarrow \infty} \rho_{\mc{E}}^{\mu} =  \lim_{\mu\rightarrow \infty}  \mc{I}\otimes \mc{E}(\Phi_{\mu}),
$ where $\rho_{\mc{E}}^{\mu} \defeq \mc{I}\otimes \mc{E}(\Phi_{\mu})$ is a finite energy quasi-Choi matrix. This treatment must be considered for any functional $f$ of asymptotic Choi matrices, such that it must be computed as the limit $f(\rho_{\mc{E}}) = \lim_{\mu\rightarrow \infty} f(\rho_{\mc{E}}^{\mu})$. This is true for the CI and RCI, and the REE. It can be shown that both hashing inequality in Eq.~(\ref{eq:HashIneq}) and the teleportation stretching upper-bound in Eq.~(\ref{eq:TSUB}) extend to general bosonic systems with constrained energy, and can extend to bosonic Gaussian channels in the limit of infinite energy (see Supplementary Notes 2, 4 of \cite{PLOB}).}.

Now consider a quantum network $\mc{N}=(P,E)$ with an optimal physical orientation of physical channels $\{\mc{E}_{\bs{xy}}^*\}_{(\bs{x},\bs{y})\in E}$. Since $I(\mc{E}_{\bs{xy}}^*)$ is an achievable rate for any channel in the network, it is possible to compute end-to-end capacity lower-bounds by supplementing the (R)CI into Eqs.~(\ref{eq:SPR}) and (\ref{eq:MPR}). That is, we can always write,
\begin{align}
{\mc{C}}^{s}(\bs{i}, {\mc{N}}) &\geq {I}^{s}(\bs{i}, {\mc{N}}) \defeq \min_C \max_{(\bs{x},\bs{y})\in \tilde{C}}  I(\mc{E}_{\bs{xy}}^*),   \label{eq:I_SPR}\\
{\mc{C}}^{m}(\bs{i}, \mc{N}) &\geq {I}^{m}(\bs{i}, {\mc{N}}) \defeq \min_C \sum_{(\bs{x},\bs{y})\in \tilde{C}} I(\mc{E}_{{\bs{xy}}}^*) \label{eq:I_MPR}.
\end{align}
Here ${I}^{s}(\bs{i}, {\mc{N}})$ and $ {I}^{m}(\bs{i}, {\mc{N}})$ provide single-path and multi-path network generalisations of ${I}(\mc{E})$ in Eq.~(\ref{eq:HashIneq}). These are achievable network rates for a quantum network in an arbitrary topology and channel composition. When all the channels $\mc{E}_{\bs{xy}}^*$ considered within the network structure are distillable then these achievable network capacities become the exact capacities.

\subsection{Upper-Bounds}

Analogously, it is always possible to derive end-to-end capacity upper-bounds by replacing the exact single-edge capacities with upper-bounds in Eqs.~(\ref{eq:SPR}) and (\ref{eq:MPR}). Let $F(\mc{E})$ be a function which computes an upper-bound on the capacity of a channel $\mc{E}$, i.e.~$\mc{C}(\mc{E}) \leq F(\mc{E})$. Then in the network setting we can write
\begin{align}
{\mc{C}}^{s}(\bs{i}, {\mc{N}}) &\leq {F}^{s}(\bs{i}, {\mc{N}}) \defeq \min_C \max_{(\bs{x},\bs{y})\in \tilde{C}}  F(\mc{E}_{\bs{xy}}^*),   \label{eq:F_SPR}\\
{\mc{C}}^{m}(\bs{i}, \mc{N}) &\leq {F}^{m}(\bs{i}, {\mc{N}}) \defeq \min_C \sum_{(\bs{x},\bs{y})\in \tilde{C}} F(\mc{E}_{{\bs{xy}}}^*) \label{eq:F_MPR}.
\end{align}
For example, teleportation-covariant channels are those which can be simulated via a teleportation protocol using their Choi matrix as a resource \cite{PLOB, TCS}, for which distillable channels are a sub-class. For teleportation-covariant channels, even when they are not distillable, it is possible to write an upper-bound on their capacity using the relative entropy of entanglement (REE) of their Choi matrix, 
\begin{equation}
{\mc{C}(\mc{E}_{\bs{xy}}^*) \leq E_{R}(\rho_{\mc{E}_{\bs{xy}}^*}) \defeq \min_{\sigma \in \mc{D}_{\text{sep}}} S(\rho_{\mc{E}_{\bs{xy}}^*}\| \sigma),} \label{eq:TSUB}
\end{equation}
where $S(\rho\|\sigma) = \text{Tr}\left[ \rho (\log\rho -\log\sigma)\right]$ is the quantum relative entropy, and the minimisation is performed over the set of all separable bipartite states $\mc{D}_{\text{sep}}$ \cite{Note2}. The REE can then be used to write end-to-end capacity bounds for networks consistent of bosonic thermal-loss channels, Pauli channels, and more.

\subsection{Node Splitting}
Ideally, quantum repeaters are completely lossless, noiseless, and fully error-corrected. Indeed, Ref.~\cite{End2End} derives end-to-end quantum network capacities under the assumption of perfect repeaters, only considering the unavoidable decoherence due to quantum channels which are external to the repeater devices. In reality, there exist a number of internal loss/noise contributions that should be considered. 
A first step in this direction was to consider the presence of loss within quantum repeaters due to
sub-optimal detection efficiency, channel-memory coupling losses, memory loading and readout \cite{RicSplitting}. In this work, we present the general scenario where repeaters are affected by both loss and noise, incorporating electronic and environmental noise affects that may occur. 

Consider a single quantum repeater contained within a quantum network, $\bs{x} \in P$. To account for internal imperfections, we may perform \textit{node splitting} \cite{RicSplitting}. A repeater node $\bs{x} \in P$ can be split into a trio of internal nodes $\bs{x} \rightarrow \{ \bs{x}^r, \bs{x}^{u}, \bs{x}^{s}\}$: A receiver node $\bs{x}^r$, a user node $\bs{x}^{u}$, and a sender node $\bs{x}^s$. The user node $\bs{x}^{u}$ represents the only valid node from which communication can originate or end, or where a user can actually be situated. Through node splitting, we can represent decoherence effects due to imperfect transmission and reception via additional quantum channels between the internal nodes.
Internal noise and loss due to imperfect reception and storage of quantum information at the node $\bs{x}$ can be described by a quantum channel physically directed from the receiver node to the user node $\mc{E}_{\bs{x}^r \rightarrow \bs{x}^u}$. 
Similarly, imperfections in the memory loading and transmission process can be captured via an internal channel between the user node and the sender node, 
$\mc{E}_{\bs{x}^u \rightarrow \bs{x}^s}$. Each internal channel will possess unique properties according to the technologies used throughout the network. 

We may consider communication between two repeater nodes $\bs{x} =  \{ \bs{x}^r, \bs{x}^{u}, \bs{x}^{s}\}$ and $\bs{y} =  \{ \bs{y}^r, \bs{y}^{u}, \bs{y}^{s}\}$ within a quantum network $\mc{N}$. We do not assume the precise nature of the external or internal channels, focussing first on a general picture. Since the internal imperfections at each node in general are unique, then the physically directed channels between $\bs{x}$ and $\bs{y}$ will be asymmetric. 
For quantum communication in the physical direction $\bs{x}\rightarrow\bs{y}$ the complete channel is a compound channel given by
\begin{align}
\mc{E}_{\bs{x}\rightarrow \bs{y}} &=  \mc{E}_{\bs{y}^r\rightarrow \bs{y}^u} \circ \mc{E}_{\bs{x}^s\rightarrow \bs{y}^r} \circ \mc{E}_{\bs{x}^u\rightarrow \bs{x}^s},
\end{align}
accounting for each channel between the user node in $\bs{x}^u \in \bs{x}$ and the user node in $\bs{y}^u \in\bs{y}$. For the physical exchange of quantum systems in the opposite direction, the compound channel reads
\begin{align}
\mc{E}_{\bs{y}\rightarrow \bs{x}} &=  \mc{E}_{\bs{x}^r\rightarrow \bs{x}^u} \circ \mc{E}_{\bs{y}^s\rightarrow \bs{x}^r} \circ \mc{E}_{\bs{y}^u\rightarrow \bs{y}^s}.
\end{align}
Applying node splitting to every edge in the network $(\bs{x},\bs{y})\in E$, and selecting the best physically directed channel for each edge in the network, we can then retrieve the optimal physical orientation of a realistic quantum network composed of imperfect repeaters, $\{ \mc{E}_{\bs{xy}}^* \}_{(\bs{x},\bs{y})\in E}$. 

\subsection{Benchmarking with Weakly-Regular Networks\label{sec:WRBench}}

In order to evaluate the end-to-end capacity bounds and node splitting procedure we need to investigate suitable network architectures. It is an open question as to how quantum networks should be best constructed on mid-to-large scales in order to balance high-rates with cost-efficient resources. Questions of this form have been recently tackled numerically via the statistical study of complex, random quantum networks \cite{BritoRandQNets, QuntaoRandQNets, ZhangQInt}. These works have been able to identify insightful phenomena of large-scale quantum networks, primarily concerned with channel length and network nodal density. 

The study of complex architectures such as Waxman, Erd\H{o}s-R\'{e}nyi and scale-free networks rely heavily upon numerical assessments, while analytical treatments of quantum repeater networks have been mostly limited to linear networks. Linear networks (or repeater chains) are effective for studying extended point-to-point communications, but are too simplistic to model large, inter-connected structures which might contain many users. To address a common ground between complex architectures and simplified repeater chains we employ WRNs as described in Section~\ref{sec:WRNs}. WRNs are capable of balancing aspects of ideality (promising, consistently high-connectivity through regularity) with reality (physically realisable with no further spatial or topological constraints). As such, they offer a versatile, analytical tool for quantum network benchmarking.

In a highly-connected network, there exist a large number of end-to-end routes between any pair of network nodes. Clearly this is desirable and  ideal for multi-path routing strategies, as end-users can simultaneously exchange quantum systems along many routes and enhance their end-to-end capacity over single-path methods. By definition, WRNs observe very consistent connective properties, while allowing substantial topological and spatial freedom. On a large-scale WRNs manifest as a realistic, well connected class of architecture which offers a very useful model for investigating the resource requirements of high-rate quantum networks.

The optimal performance of WRNs is surprisingly predictable and relates to the concept of \textit{network cut growth}. The end-to-end flooding capacity in Eq.~(\ref{eq:MPR}) is characterised by a cut $C_{\min}$ which minimises the multi-edge capacity across all possible cuts on the network. In a WRN (and other highly-connected networks) the existence of many end-to-end routes between end-users requires the collection of many more edges in a cut-set $\tilde{C}$ to successfully partition them. Interestingly, a relationship between cut-set cardinality $|\tilde{C}|$ and distance from an end-user then emerges; as a network becomes more highly-connected, trying to perform cuts using edges that are further away from $\bs{\alpha}$ or $\bs{\beta}$ requires the collection of many more edges (i.e.~cut growth).

Due to network cut growth it becomes increasingly likely that the minimum cut $C_{\min}$ will be one which collects a relatively small number of edges, i.e.~it is easier to minimise the multi-edge capacity if there are fewer edges to sum. By this logic, there should exist some minimum single-edge capacity for any edge $(\bs{x},\bs{y})\in E$, $\mc{C}_{\bs{xy}} \geq \mc{C}_{\min}$, so that the minimum cut is \textit{exactly} the cut-set which collects the fewest number of edges. If the network cut growth of a particular network model can be characterised it is possible to derive $\mc{C}_{\min}$ and reveal invaluable information for quantum network design. 

Crucially, WRNs admit an analytical form which makes it possible to compute how cut-sets grow in cardinality with respect to increasing distance from an end-user. Using this information in conjunction with the upper-bound in Eq.~(\ref{eq:FloodUB}), it is possible to derive conditions for which optimal flooding performance is guaranteed. More precisely, WRNs allow us to derive \textit{threshold theorems}; theorems which reveal threshold single-edge capacities $\mc{C}_{\min}$ that guarantee specific end-to-end performance bounds. 

Ref.~\cite{OPGQN} introduced these threshold theorems and applied them to bosonic lossy WRNs. In this work, we generalise threshold theorems to study the optimal performance of WRNs constituted of any quantum channel, even if their exact capacity is not known. Using the end-to-end capacity bounds from Section~\ref{sec:GenBounds} and the technique of node splitting, we are able to derive bounds on critical parameters such as maximum channel length, maximum internal loss or maximum thermal noise which are able guarantee optimal performance. Appendix~\ref{sec:ThreshTheorems} collects these theorems which are put to use in the following section.

\section{Applications}

Using the tools of Eqs.~(\ref{eq:I_SPR})-(\ref{eq:F_MPR}) we benchmark the optimal performance of quantum networks with imperfect repeaters. We incorporate realistic, noisy channels throughout the network for which the exact capacities are not known. This allows us to investigate the end-to-end rates of lesser studied network models such as amplitude-damping networks and bosonic thermal-loss networks.

\subsection{Amplitude Damping Networks}
\subsubsection{General Bounds}
We begin with networks consistent of amplitude damping (AD) quantum channels. AD channels are qubit channels that describe the process of energy dissipation through spontaneous emission, representing the analogue of a bosonic lossy channel constrained to a two-level system. They are ubiquitous in the modelling of many important physical processes in communications and computation. For damping-probability (or loss) $p \in (0,1)$, the AD channel can be defined via the Kraus decomposition,
\begin{gather}
\mc{E}_p (\rho) = \sum_{i=0}^1 K_i \rho K_i^{\dag}, \\
K_0 \defeq \ket{0}\!\bra{0} + \sqrt{1-p} \ket{1}\!\bra{1}, \> \> K_1 \defeq \sqrt{p} \ket{0}\!\bra{1}.
\end{gather}
Interestingly, the AD channel is not distillable, hence its exact capacity is unknown.

Let us consider an arbitrary quantum network composed of AD channels, $\mc{N}_{\text{AD}}$. For communication between two network nodes  ${\bs{x} =  \{ \bs{x}^r, \bs{x}^{u}, \bs{x}^{s}\}}$ and $\bs{y} =  \{ \bs{y}^r, \bs{y}^{u}, \bs{y}^{s}\}$, we can describe each of the internal and external channels as AD channels with a unique damping probability. 
More precisely, let any external channel between two nodes $\bs{x}$ and $\bs{y}$ be AD channels with the damping probability
$p_{\bs{xy}} = p_{\bs{x^s \rightarrow y^r}} = p_{\bs{y^s \rightarrow x^r}}$. The internal channels can also be considered AD channels with fixed physical directions,
\begin{equation}
\mc{E}_{\bs{i}^r \rightarrow \bs{i}^u} = \mc{E}_{p_{\bs{i}}^r },~~\mc{E}_{\bs{i}^u \rightarrow \bs{i}^s} = \mc{E}_{p_{\bs{i}}^s }, ~\bs{i}\in\{\bs{x},\bs{y}\}.
\end{equation} 
For the physical exchange of quantum systems in either direction, the complete compound channels read
\begin{align}
\mc{E}_{\bs{x}\rightarrow \bs{y}}
&= \mc{E}_{p_{\bs{y}}^r}\circ \mc{E}_{{p}_{\bs{xy}}} \circ \mc{E}_{p_{\bs{x}}^s} \label{eq:ForAD},\\
\mc{E}_{\bs{y}\rightarrow \bs{x}} 
&= \mc{E}_{p_{\bs{x}}^r}\circ \mc{E}_{{p}_{\bs{xy}}} \circ \mc{E}_{p_{\bs{y}}^s}.\label{eq:BackAD}
\end{align}

For any quantum network topology, we can use these compound channels and their capacities to assign an optimal physical orientation and regain an undirected graphical representation. To do this, the RCI places an achievable lower-bound on the single-edge capacity of each compound channel throughout the network. For AD channels, this results in the following achievable rate for each network edge \cite{PirPatron09},
\begin{equation}
 I(\mc{E}_{\bs{xy}}^*) =\hspace{-2mm} \max_{\bs{i} \neq \bs{j} \in \{\bs{x},\bs{y}\}} \max_{u} \left[ H_2(u) - H_2(up_{\bs{i}\rightarrow\bs{j}}^{\text{tot}} )\right]. \label{eq:AD_LB}
\end{equation}
Here, $H_2(u) \defeq -u\log_2(u) - (1-u)\log_2(1-u)$ defines the binary Shannon entropy function, and $p_{\bs{i}\rightarrow\bs{j}}^{\text{tot}}$ is the total damping probability associated with a compound channel directed from node $\bs{i}$ to $\bs{j}$,
\begin{align}
p_{\bs{i}\rightarrow\bs{j}}^{\text{tot}} &\defeq 1 - (1-p_{\bs{i\rightarrow j}}^{\text{int}})(1-p_{\bs{ij}}),\\
&= 1 - (1-p_{\bs{j}}^r)(1-p_{\bs{ij}})(1-p_{\bs{i}}^s), 
\end{align}
where we simultaneously define an effective internal loss parameter $p_{\bs{i\rightarrow j}}^{\text{int}} \defeq p_{\bs{i}}^{r} (1-p_{\bs{j}}^s) + p_{\bs{j}}^s$ which captures both sending and receiving inefficiencies. Eq.~(\ref{eq:AD_LB}) can be substituted into the single-path and multi-path expressions in Eqs.~(\ref{eq:I_SPR}) and (\ref{eq:I_MPR}) to provide achievable network rates. 

Upper-bounds on the single-edge capacity of AD channels can also be expanded to provide upper-bounds on the end-to-end capacities. 
One of the tools that can be used for AD channels is the squashed entanglement \cite{ChristandlSqE,TakeokaSqE, AzumaRLT,ToolsNetDesign}. Ref.~\cite{PLOB} showed that,
\begin{equation}
{E}_{\text{sq}}(\mc{E}_{\bs{xy}}^*) = \hspace{-3mm}\max_{\bs{i} \neq \bs{j} \in \{\bs{x},\bs{y}\}} \hspace{-2mm} H_2\bigg(\frac{1}{2}-\frac{p_{\bs{i}\rightarrow\bs{j}}^{\text{tot}}}{4} \bigg) - H_2\bigg( 1 - \frac{p_{\bs{i}\rightarrow\bs{j}}^{\text{tot}}}{4} \bigg).\label{eq:AD_UB}
\end{equation}
This can then be used within Eqs.~(\ref{eq:F_SPR}) and (\ref{eq:F_MPR}) to compute upper-bounds on the ultimate limits of quantum communications over AD networks.

\begin{figure*}
\includegraphics[width=\linewidth]{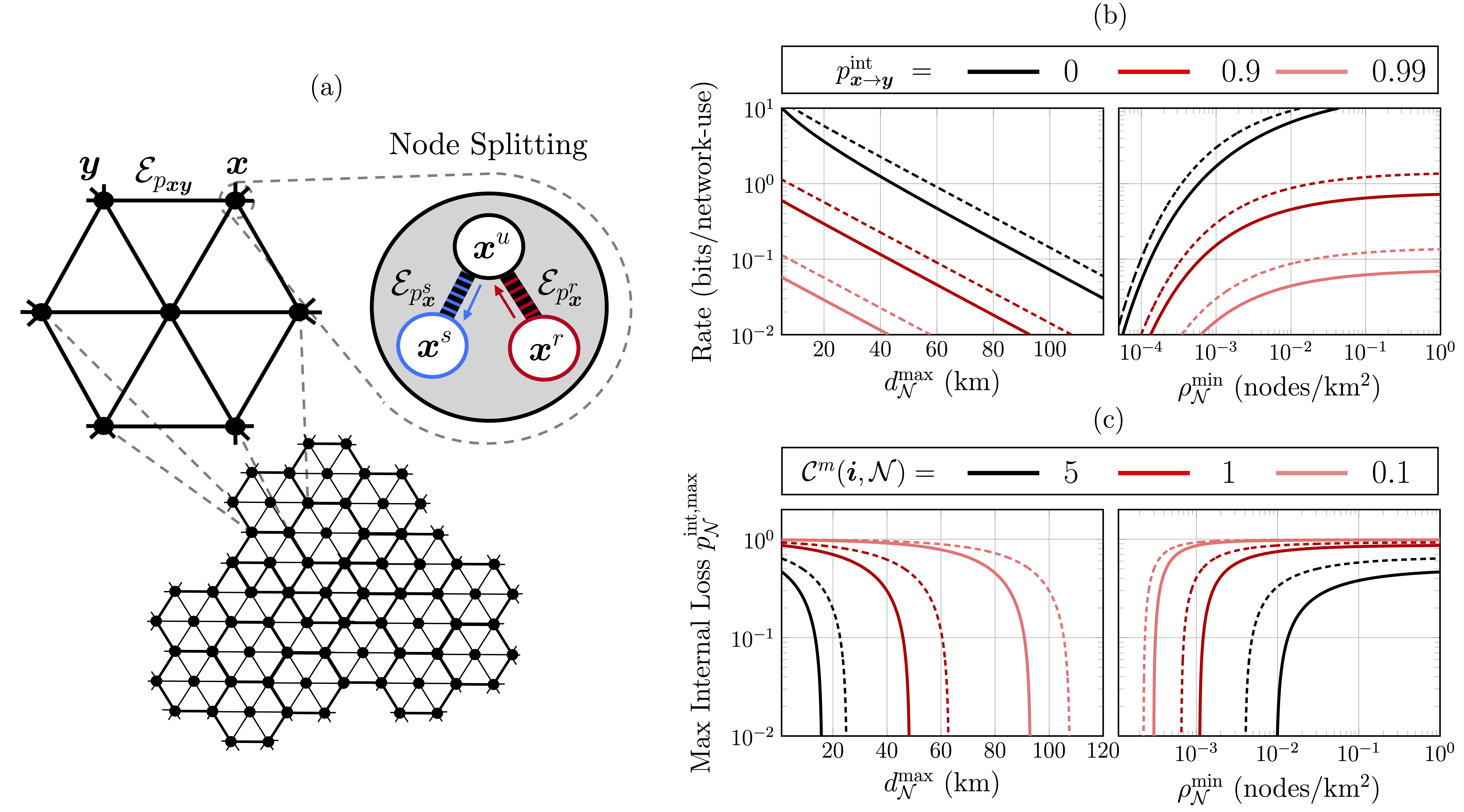}
\caption{(a) Node-splitting procedure for quantum amplitude-damping networks. We consider quantum WRNs within some nodal boundary which satisfy regularity $k=6$ and adjacent commonality $\bs{\lambda}_{\bs{x}} = \{2\}^{\cup 6}$ for any node $\bs{x} \in P$.
Throughout all plots, dashed lines represent bounds obtained using the squashed entanglement as a single-edge capacity upper-bound. Meanwhile, all solid lines plot bounds derived using the RCI as a single-edge capacity lower-bound.
Panel (b) plots bounds on the maximum fibre-length permitted within the network $d_{\mc{N}}^{\max}$ and a corresponding minimum nodal density $\rho_{\mc{N}}^{\min}$ which ensure that an end-user pair can obtain optimal performance $\mc{C}^m(\bs{i},\mc{N})$. Panel (c) depicts bounds on the maximum tolerable internal loss $p_{\mc{N}}^{\max}$ permitted within the network with respect to maximum fibre-length and nodal density so to guarantee an optimal flooding capacity.}
\label{fig:ADWRN}
\end{figure*}

\subsubsection{Example Network Model}

To make use of these bounds and the tools of Section~\ref{sec:WRBench}, we can investigate the end-to-end performance of weakly-regular AD networks. In Fig.~\ref{fig:ADWRN} we consider a $k=6$ weakly-regular network, which adopts equivalent connectivity properties to a triangular lattice. In Fig.~\ref{fig:ADWRN}(a) a single cell of this architecture is illustrated, and it is shown how a $k=8$ WRN can be constructed. We investigate large networks consisting of many of these cells connected together, and consider end-users which are deeply-embedded in the network so that any network boundary effects can be ignored \footnote{This is a very reasonable assumption, and provides a better picture of end-to-end performance within a large network setting. Formal requirements for this has been established in Ref.~\cite{OPGQN}, and we call such networks \textit{internally weakly-regular.}}. 

All edges $(\bs{x},\bs{y})\in E$ in the network are modelled by AD channels of length $d_{\bs{xy}}$, used to describe optical-fibre. As such, the damping-probability of each edge is given by
\begin{equation}
p_{\bs{xy}} = 1 - 10^{-\gamma d_{\bs{xy}}},
\end{equation}
where $\gamma = 0.02 \text{ per km}$ is the state-of-the-art loss-rate for fibre ($0.2$ dB per km). Node-splitting is applied throughout the network in order to incorporate internal loss, as depicted in Fig.~\ref{fig:ADWRN}(a). By specifying a threshold theorem to the property of channel length, we are able to place tight bounds on the \textit{maximum tolerable fibre-length} permitted within the network, $d_{\mc{N}}^{\max}$, with respect to a desired flooding capacity. That is, for a given flooding capacity between a pair of end-users $\mc{C}^m(\bs{i},\mc{N})$ and fixed internal losses $p_{\bs{x}\rightarrow\bs{y}}^{\text{int}}$, we plot the maximum fibre-length that is allowed within the network-bulk so that we can guarantee the flooding capacity is optimal. 

Thanks to the consistent, analytical connectivity properties of WRNs, the maximum link-length can also be used to derive a \textit{minimum nodal density}, $\rho_{\mc{N}}^{\min}$. Given a WRN with consistent connectivity rules and a maximum link-length, the minimum nodal density presents a lower-bound on the number of nodes per unit area when defined over a spatial area. For analytical structures such as WRNs, a relationship between $d_{\mc{N}}^{\max}$ and $\rho_{\mc{N}}^{\min}$ can be identified by determining a least dense configuration of the architecture. Using this relationship, it is possible to extend the link-length threshold theorem to place bounds on the nodal density requirements of a quantum WRN necessary to guarantee high performance. For the $k=6$ WRN considered here it can be shown that
$
\rho_{\mc{N}} \geq \rho_{\mc{N}}^{\min} \geq \frac{2}{\sqrt{3}} (d_{\mc{N}}^{\max})^{-2}.
$ See Appendix~\ref{sec:NodalDens} for more details.

\subsubsection{Analysis}
 
Fig.~\ref{fig:ADWRN}(b) emphasises that end-to-end performance has a clear and obvious dependence on internal decoherence; repeater inefficiencies ultimately limit end-to-end performance, and as a result place stricter constraints on the maximum permitted fibre-length and resources required throughout the network. For lossy network nodes with a total internal efficiency of $p_{\bs{x}\rightarrow\bs{y}}^{\text{int}} = 0.9$ (10\% efficient) at any node require that fibre-lengths are limited to approximately $100\text{ km}$ in order to achieve a flooding capacity of $\mc{C}^m(\bs{i},\mc{N}) = 10^{-2}$ bits per network use. This corresponds to a minimum nodal density of approximately $\rho_{\mc{N}}^{\min} \approx 1\times10^{-4}$ nodes per km${}^2$.

Fig.~\ref{fig:ADWRN}(c) explores this relationship further. Given some end-to-end optimal performance $\mc{C}^{m}(\bs{i},\mc{N})$, we plot the maximum tolerable internal loss $p_{\mc{N}}^{\text{int},\max}$ permitted at each node the network, with respect to the maximum fibre-length (and minimum nodal density) necessary to achieve it. This elucidates the \textit{required} efficiency of repeater stations throughout the network, given some maximum fibre-length and desired end-to-end capacity. We see clearly that for very high-rates $\mc{C}^m(\bs{i},\mc{N}) = 5$, the maximum fibre-length must be limited to approximately $15 - 25 \text{ km}$, otherwise the tolerable internal loss tends to zero, i.e.~it can only be achieved via perfect devices. Similarly, the corresponding minimum nodal density in this setting is very large, requiring on the order of $\sim 10^{-2}$ nodes per km${}^2$. This is the expected resource requirements for high-rate DV quantum communications within a metropolitan setting. For lower rates, each node can tolerate greater inefficiencies over longer channel lengths. Nonetheless, to achieve a flooding rate of $\mc{C}^m(\bs{i},\mc{N}) = 0.1$, channel lengths must be limited to below $93\text{ km}$ in the worst-case or $107\text{ km}$ in the best-case  \footnote{The range of values corresponds to worst-case and best-case bounds on the threshold parameter using the upper or lower-bounds on the end-to-end capacity}.

\subsection{Bosonic Thermal-Loss Networks}

\subsubsection{General Bounds}

 For bosonic quantum communications, one of the most important channels is the Gaussian thermal-loss channel $\mc{E}_{\eta,\bar{n}}$ of transmissivity $\eta \in (0,1)$ and output thermal noise of $\bar{n}$ photons. This channel can be described by the action of a beam-splitter of transmissivity $\eta$ which mixes the input mode with an environmental thermal mode with mean photon number
 \begin{equation}
 \bar{n}_{\text{env}} \defeq \bar{n}/(1-\eta).
 \end{equation}
 This transforms the input quadratures of a single-mode input Gaussian state ${\hat{\bff{x}} = (\hat{q},\hat{p})^T}$ according to $\hat{\bff{x}} \rightarrow \sqrt{\eta}\hat{\bff{x}} + \sqrt{1-\eta}\hat{\bff{x}}_{\text{env}}$, where $\hat{\bff{x}}_{\text{env}}$ is the quadrature operator of the thermal environmental mode. When $\bar{n} = 0$, then the environmental mode is in the vacuum state and this becomes a pure-loss (or lossy) channel $\mc{E}_{\eta}$. CV quantum communication through optical-fibres or free-space are most accurately modelled using thermal-loss channels. Similarly, thermal-loss channels can most effectively model internal noise/losses within realistic CV quantum repeaters.

Let us consider a general thermal-loss quantum network $\mc{N}_{\text{TL}}$. We will assume that two repeater nodes ${\bs{x} =  \{ \bs{x}^r, \bs{x}^{u}, \bs{x}^{s}\}}$ and $\bs{y} =  \{ \bs{y}^r, \bs{y}^{u}, \bs{y}^{s}\}$ are connected via a thermal-loss channel with the attenuation and thermal noise properties
\begin{equation}
\eta_{\bs{xy}} = \eta_{\bs{x^s \rightarrow y^r}} = \eta_{\bs{y^s \rightarrow x^r}},~~ \bar{n}_{\bs{xy}} = \bar{n}_{\bs{x^s \rightarrow  y^r}} = \bar{n}_{\bs{y^s \rightarrow  x^r}}.
\end{equation} 
Furthermore, let us model the internal repeater channels as thermal-loss channels which have fixed physical directions,
\begin{equation}
\mc{E}_{\bs{i}^r \rightarrow \bs{i}^u} = \mc{E}_{{\tau}_{\bs{i}}^r, \bar{n}_{\bs{i}}^r },~~ \mc{E}_{\bs{i}^u \rightarrow \bs{i}^s} = \mc{E}_{{\tau}_{\bs{i}}^s, \bar{n}_{\bs{i}}^s }, ~\bs{i}\in\{\bs{x},\bs{y}\}.
\end{equation} 
Hence, the complete compound channels in either physical direction are
\begin{align}
\mc{E}_{\bs{x}\rightarrow \bs{y}} 
&= \mc{E}_{\tau_{\bs{y}}^r, \bar{n}_{\bs{y}}^r}\circ \mc{E}_{{\eta}_{\bs{xy}},\bar{n}_{\bs{xy}}} \circ \mc{E}_{\bar{n}_{\bs{x}}^s,\tau_{\bs{x}}^s}, \label{eq:ForTL}\\
\mc{E}_{\bs{y}\rightarrow \bs{x}} 
&= \mc{E}_{\tau_{\bs{x}}^r, \bar{n}_{\bs{x}}^r}\circ \mc{E}_{{\eta}_{\bs{xy}},\bar{n}_{\bs{xy}}} \circ \mc{E}_{\bar{n}_{\bs{y}}^s,\tau_{\bs{y}}^s}. \label{eq:BackTL}
\end{align}

In general these compound channels are physically asymmetric. Thus for any quantum network topology, we need assign an optimal physical orientation to regain an undirected graphical representation. Once again, this can be achieved using the RCI to lower-bound the capacity of any compound thermal-loss channel in the network \cite{PirPatron09}. More precisely, for the forward channel we compute,
\begin{align}
I(\mc{E}_{\bs{x}\rightarrow \bs{y}}) 
&= -\log_2(1-\eta_{\bs{x}\rightarrow\bs{y}}^{\text{tot}}) - h\left(\frac{\bar{n}_{{\bs{x}\rightarrow\bs{y}}}^{\text{tot}}}{1-\eta_{\bs{x}\rightarrow\bs{y}}^{\text{tot}}}\right), \label{eq:TL_LB}
\end{align}
where we have made use of the entropic function $
{h(x) \defeq (x+1)\log_2(x+1) - x\log_2(x)},
$
and the total point-to-point transmissivity and thermal noise parameters
\begin{align}
&\eta_{\bs{x}\rightarrow\bs{y}}^{\text{tot}} \defeq {\tau_{\bs{y}}^r} {\tau_{\bs{x}}^s}{\eta}_{\bs{xy}} \label{eq:TotalTran}, \\
&\bar{n}_{\bs{x}\rightarrow\bs{y}}^{\text{tot}} \defeq  \bar{n}_{\bs{y}}^r + \tau_{\bs{y}}^r \bar{n}_{\bs{xy}} + \eta_{\bs{xy}} \tau_{\bs{y}}^r \bar{n}_{\bs{x}}^s. \label{eq:xi_N}
\end{align}
These total parameters are derived by finding a single channel representation of compound thermal-loss channels (see Appendix~\ref{sec:CompDer} for more details). For the backward channel we can simply reverse the order of the nodal directions in Eqs.~(\ref{eq:TotalTran}) and (\ref{eq:xi_N}), retrieving $I(\mc{E}_{\bs{y}\rightarrow \bs{x}})$.
These offer lower-bounds on the capacity of each compound channel. By comparing these quantities, we can then identify the optimal physical channel $\mc{E}_{\bs{xy}}^{*}$ for each edge in the network whose capacity is lower-bounded by,
\begin{equation}
I(\mc{E}_{\bs{xy}}^*) = \max_{\bs{i}\neq\bs{j}\in\{\bs{x},\bs{y}\}} I(\mc{E}_{\bs{i}\rightarrow\bs{j}}).
\end{equation}
This capacity lower-bound can then be substituted into Eq.~(\ref{eq:I_SPR}) and Eq.~(\ref{eq:I_MPR}) for achievable end-to-end rates.

Thermal-loss channels are teleportation-covariant, but are not distillable. Hence, we can compute upper-bounds on the network capacities using the REE single-edge upper-bounds \cite{PLOB}. For either physically directed channel this upper-bound takes the form,
\begin{align}
E_R(\mc{E}_{\bs{i}\rightarrow\bs{j}}) = I(\mc{E}_{\bs{i}\rightarrow\bs{j}}) - \frac{\bar{n}_{{\bs{x}\rightarrow\bs{y}}}^{\text{tot}}}{1-\eta_{\bs{x}\rightarrow\bs{y}}^{\text{tot}}} \log_2 (\eta_{\bs{i}\rightarrow\bs{j}}^{\text{tot}}).
\end{align}
where we have used the RCI from Eq.~(\ref{eq:TL_LB}).
By optimising the physical orientation as before, we can substitute
\begin{align}
E_R({\mc{E}_{\bs{xy}}^*}) = \max_{\bs{i}\neq\bs{j}\in\{\bs{x},\bs{y}\}} E_R(\mc{E}_{\bs{i}\rightarrow\bs{j}}).
\end{align}
into Eq.~(\ref{eq:F_SPR}) and  Eq.~(\ref{eq:F_MPR}) to produce network capacity upper-bounds. Since this is an upper-bound, it is not known if they are achievable rates. However, using $I({\mc{E}_{\bs{xy}}^*})$ and $E_R({\mc{E}_{\bs{xy}}^*})$ it is possible to accurately bound the end-to-end capacities.

\subsubsection{Example Network Model}

\begin{figure*}
\includegraphics[width=\linewidth]{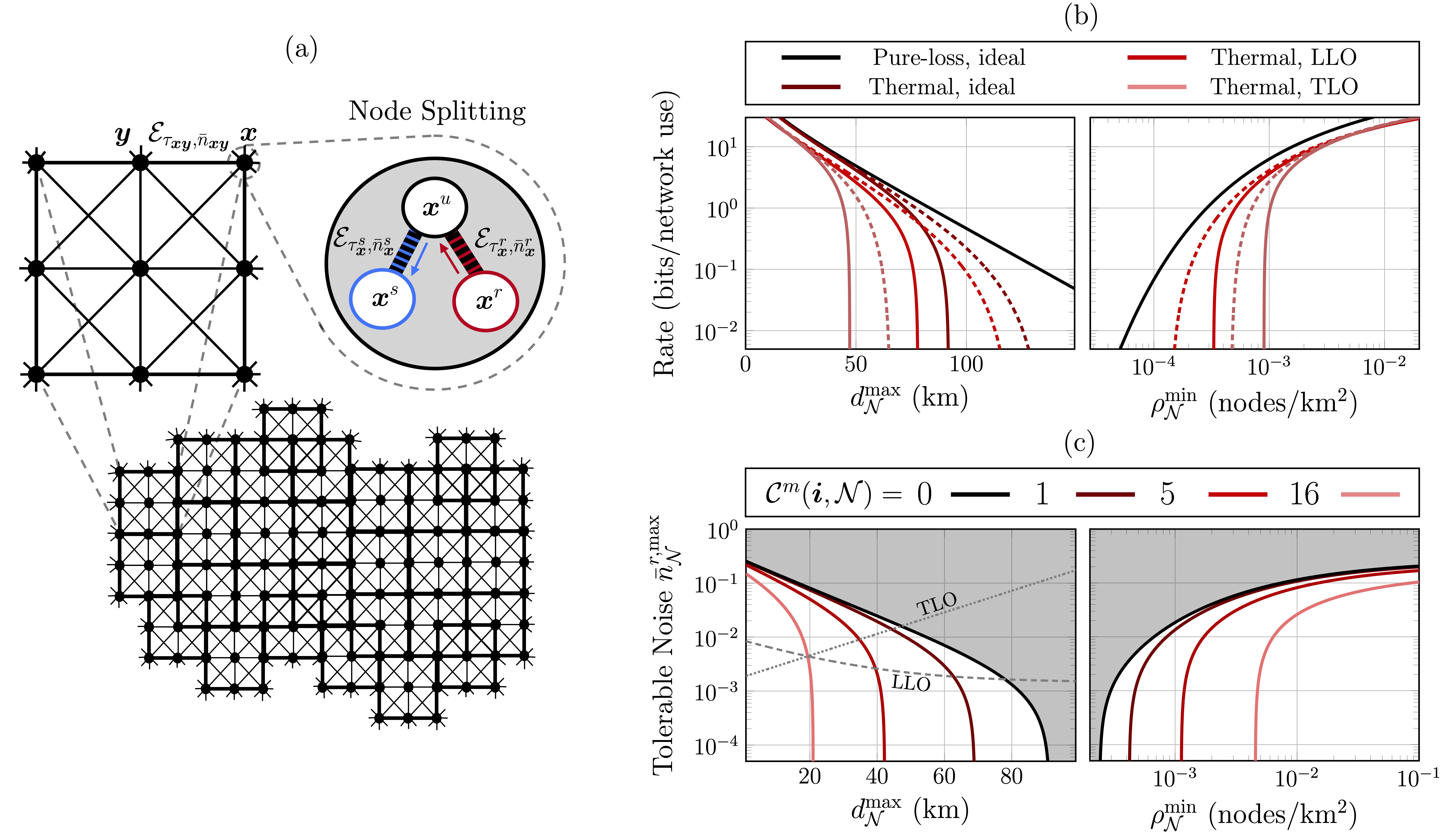}
\caption{(a) Node-splitting procedure for bosonic thermal-loss quantum networks. We consider quantum WRNs within some nodal boundary which satisfy regularity $k = 8$ and adjacent commonality $\bs{\lambda}_{\bs{x}} = \{2,4\}^{\cup 4}$ for any node $\bs{x}\in P$. 
Throughout all plots, dashed lines represent bounds obtained using the REE as a single-edge capacity upper-bound. Meanwhile, all solid lines plot bounds derived using the RCI as a single-edge capacity lower-bound.
Panel (b) plots the maximum inter-nodal separation $d_{\mc{N}}^{\max}$ permitted within the network, and its associated minimum nodal density, such that optimal performance $\mc{C}^m(\bs{i},\mc{N})$ can be obtained. Here, we consider network setups consisting of ideal repeaters or imperfect repeaters using LLOs and TLOs and heterodyne detection. Panel (c) displays the relationship between maximum fibre-length, minimum nodal density and tolerable thermal noise at the receiver $\bar{n}_{\mc{N}}^{r,\max}$ throughout the network. The grey areas of each plot illustrate regions of network parameter space for which we are unable to guarantee \textit{any} end-to-end capacity whatsoever; identifying essential properties for realistic thermal-loss networks.  }
\label{fig:WR_TL}
\end{figure*}

We can now benchmark the end-to-end limits of a bosonic thermal-loss network with imperfect repeaters. Here, we consider a $k=8$ weakly-regular network, inspired by a Manhattan-like structure. A single network cell is depicted in Fig.~\ref{fig:WR_TL}(a) which is used to construct larger designs. Recall that we only demand connectivity constraints, and place no requirements on the spatial or topological properties of the network. The node-splitting procedure is similarly performed, adopting lossy and noisy channels between internal nodes to capture repeater inefficiencies. Once again, we assume that end-users are located within some nodal boundary where boundary effects are unimportant \cite{Note2}.

We assume all network edges $(\bs{x},\bs{y})\in E$ are thermal-loss channels used to model optical-fibre so that the transmissivity is given by ${\eta_{\bs{xy}} = 10^{-{\gamma d_{\bs{xy}}}}}$ using the fibre-loss rate $\gamma = 0.02$ as before. Furthermore, there is unavoidable background thermal noise ${\bar{n}_B}$ which is added to the propagating mode through the fibre-channel for which we assume the typical value of $\bar{n}_B \approx 0.002$ in our numerical investigations. As a result, any external fibre-channel of length $d_{\bs{xy}}$ in a bosonic CV quantum network can be modelled as a thermal-loss channel with these parameters.
As explored with regards to AD networks, we are able to utilise threshold theorems to derive maximum fibre-lengths $d_{\mc{N}}^{\max}$ necessary to guarantee optimal performance bounds. Similarly, this information can be used to identify network nodal density requirements using the minimum nodal density. For the WR architecture considered here, these quantities are connected via $\rho_{\mc{N}}^{\min} \geq 2/(d_{\mc{N}}^{\max})^2$ \cite{OPGQN}.

Thanks to the node-splitting technique, we can also incorporate internal decoherence. In particular, we can investigate the end-to-end performance limits of realistic CV-QKD networks by considering setup noise/loss introduced by specific protocols. Typical CV protocols will make use of either homodyne or heterodyne measurements, both of which rely upon the use of a local-oscillator (or phase reference). This contains information that allows the sender and receiver to exploit both quadratures of the mode. The local-oscillator can be established via two key methods: The transmitted local-oscillator (TLO) or the local local-oscillator (LLO).
A TLO is an additional mode which is co-propagated along with the signal-mode from the sender to receiver, carrying the relevant phase information \cite{RalphCV,AdvCrypt}. Alternatively, the LLO method interleaves signal pulses with bright reference pulses to reconstruct the phase reference locally at the receiver \cite{LLO, LLO2}. Inevitably, both techniques result in setup noise and loss experienced at the receiver. By considering internal noise and loss sources alongside the external contributions from the fibre channel, we can gain insight into the realistic limits of CV-QKD networks which rely upon these techniques (see Appendix \ref{sec:LO_App} for more details).

\subsubsection{Analysis}

In Fig.~\ref{fig:WR_TL}(b) we plot bounds on the maximum fibre-length, and corresponding minimum nodal densities of bosonic thermal-loss networks necessary to guarantee an optimal flooding capacity $\mc{C}^m(\bs{i},\mc{N})$ between an end-user pair $\bs{i}=\{\bs{\alpha},\bs{\beta}\}$. We do this for a number of setups; a pure-loss network with perfect repeaters, a thermal-loss network with perfect repeaters, and thermal-loss networks with imperfect repeaters using either TLOs/LLOs with heterodyne measurements (see Appendix \ref{sec:LO_App} and Table~\ref{table:Setups} for specific details and setup parameters).

For pure-loss networks with ideal repeaters, the maximum inter-nodal separation can become very large as we loosen our demands on the flooding capacity. 
Indeed, in this ideal scenario ${d_{\mc{N}}^{\max} \approx 183\text{ km}}$ for end-to-end rates of $\mc{C}^m(\bs{i},\mc{N}) = 10^{-2}$ bits per network use. However, upon realistic consideration of thermal noise, we realise that this upper-bound is very optimistic. Factoring background noise along each edge, is clear that $d_{\mc{N}}^{\max}$ saturates within some limiting range. For rates on the order of $\mc{C}^m(\bs{i},\mc{N}) = 10^{-2}$ bits per network use, the maximum inter-nodal separation is found in the interval $d_{\mc{N}}^{\max} \in [91, 126] \text{ km}$ \cite{Note3}. This is clearly much stricter than what is predicted when only external loss is considered. 

Consideration of internal imperfections causes $d_{\mc{N}}^{\max}$ to become even stricter. In Fig.~\ref{fig:WR_TL}(b) we study the impact that the use of practical CV-QKD setups (using heterodyne detection) have on network resource requirements. It can be seen that LLO based protocols are superior in maintaining a larger tolerable channel length throughout a network, compared to a TLO approach. While LLOs introduce phase errors, the electronic noise imparted as the receiver is independent from channel transmissivity and thus channel length. As a result, it does not degrade the tolerable channel length, unlike TLO based protocols. Indeed, for end-to-end rates of $\mc{C}^m(\bs{i},\mc{N}) = 10^{-2}$ bits per network use, CV-QKD protocols which utilise LLOs can tolerate at least $\sim 29 \text{ km}$ additional channel length for every point to point link (at most $\sim 48 \text{ km}$ by considering upper-bounds).
 
Stricter link-length demands have a substantial impact on the nodal density requirements of bosonic thermal-loss networks. This is clear within our results; when the maximum fibre-length saturates, so too does the minimum nodal density. Indeed, the consideration of thermal noise within practical CV-QKD networks demands that the minimum nodal density always be at least of order $10^{-4}$ nodes per km${}^2$ in order to achieve any non-zero end-to-end rate. This has significant ramifications on the resource requirements of realistic bosonic quantum networks.

A threshold theorem can also be derived with respect to thermal noise at the receiver $\bar{n}_{\bs{x}}^r$, so that we can further study the interplay between maximum channel length permitted in the WRN and tolerable receiver noise. Fig.~\ref{fig:WR_TL}(c) illustrates this relationship for a number of desired flooding capacities. Here we plot lower-bounds on the maximum tolerable noise $\bar{n}_{\mc{N}}^{r,\max}$ for any node in the network, using the RCI based bound. This reveals a permissible region of network parameters for which not only optimal performance is guaranteed, but non-zero rates are guaranteed. In order to guarantee a non-zero end-to-end capacity (even within this highly-connected architecture) channel lengths should be kept below $91\text{ km}$ at worst, and $126\text{ km}$ at best. Fibre-networks in this configuration which have channels longer than this are not guaranteed to have a non-zero rate. Globally, these channel length constraints manifest in the required nodal density so that we can identify a permissible region of nodal densities to guarantee non-zero rates.

Furthermore, we compare the tolerable noise bounds to the actual noise properties of CV-QKD protocols using TLOs and LLOs. It is once more illustrated that LLOs are much more effective for use within large scale networks, as the associated internal noise scales much more favourably with increasing fibre-lengths.

\section{Conclusion\label{sec:Conc}}

We have presented general bounds for the end-to-end capacities of arbitrary quantum networks. This includes achievable lower-bounds based on the coherent/reverse coherent information which apply to networks composed of any type of quantum channel. We also show how upper-bounds can be obtained using appropriate single-edge capacity bounding functions. 
Employing these bounds in conjunction with a recently developed node-splitting technique, we reveal ways to bound the end-to-end capacity of quantum networks with lossy and noisy repeaters. As a result, we provide versatile tools to investigate the internal and external decoherence properties of realistic quantum networks. 

Making use of these general results, we apply the node-splitting technique to qubit amplitude-damping networks, and bosonic thermal-loss networks; channel models for which the single-edge capacity is not exactly known. Using the class of highly-connected, weakly-regular quantum networks we are able to illuminate critical network properties upon which high-rate quantum communications rely. This allows us to identify internal loss and thermal noise thresholds permitted within quantum repeaters which guarantee optimal end-to-end performance.  

Our results find valuable insight for the infrastructure requirements of future quantum networks. Most prominently, we emphasise the necessity for considering both internal and external thermal noise when designing quantum architectures, as these noise sources can severely restrict the ability to use long fibre-channels while maintaining high-rates. Even when channels are limited in length, quantum repeaters insufficiently protected from noise will compromise performance. Future investigative paths will aim to exploit these bounds to both study and motivate realistic and high-rate network designs, and unveil the resiliency of quantum networks to unavoidable thermal noise. 

\acknowledgements
C.H acknowledges funding from the EPSRC via a Doctoral Training Partnership (EP/R513386/1). S.P acknowledges funding by the European Union via “Continuous Variable Quantum Communications” (CiViQ, Grant Agreement No. 820466).\\

%\bibliography{RCI_Nets}


\appendix

\section{Network Parameter Benchmarking with Weakly-Regular Networks\label{sec:ThreshTheorems}}

The invaluable mathematical tool within this work is the theory of threshold theorems for WRNs. These are theorems which utilise the connectivity properties of  $(k,\bs{\Lambda})$-WRNs (as introduced and defined in Section~\ref{sec:WRNs}) in order to derive \textit{single-edge threshold conditions} necessary to guarantee strong end-to-end communication performance. These single-edge threshold conditions define upper or lower bounds on physical properties of any single node or channels within the network, e.g.~channel length, thermal noise at a receiver node, internal loss, and more. Such threshold values help us to understand the resilience of quantum networks given some desirable level of performance, helping to benchmark the requirements of realistic quantum technologies necessary to perform at high-rates. 

The derivation these theorems relies upon the analytical description of WRNs and the ability to understand network cut growth in these models. For more background and details we refer the reader to Ref.~\cite{OPGQN} and its Supplementary Material where these concepts were originally derived. In this appendix, we first briefly revise the original threshold theorem for WRNs purely in the context of exact single-edge capacities. We then extend this notion to networks composed of quantum channels whose exact capacities depend upon physical properties, developing a general framework for computing network threshold parameters. Finally, this is extended once more to the situation where exact single-edge capacity expressions are not known but for which we possess bounding functions. 

\subsection{Threshold Theorems with Exact Capacities}

First, we revise threshold theorems with respect to exact single-edge channel capacities. This is the same as answering the following question: Given an end-user pair $\bs{i}=\{\bs{a},\bs{b}\}$, what is the minimum single-edge capacity in a quantum WRN necessary to ensure that the end-to-end flooding capacity is optimally achieved by the min-neighbourhood capacity, $\mc{C}_{\mc{N}_{\bs{i}}}^m$? Originally solved in Ref.~\cite{OPGQN}, we restate the answer in the following.

\begin{theorem}
Consider a $(k,\bs{\Lambda})$-WR quantum network. Select an end-user pair $\bs{i} = \{\bs{a},\bs{b}\}$, and demand they are sufficiently distant such that they do not share an edge. Then there exists a threshold single-edge capacity $\mc{C}_{\min}$ in the network, given by
\begin{gather}
\mc{C}_{\min} \defeq \frac{1}{\delta}\>\mc{C}_{\mc{N}_{\bs{i}}}^m, \label{eq:Thresh}
\end{gather}
where $\delta$ is a characteristic property of the network, 
\begin{equation}
{\delta} \defeq \min_{\bs{\lambda}\in \bs{\Lambda}}{\sum_{\lambda\in\bs{\lambda}} k - \lambda - 1},
\end{equation}
such that if all single-edge capacities in the network satisfy this minimum threshold,
$\mc{C}_{\bs{xy}} \geq \mc{C}_{\min},\forall\>(\bs{x},\bs{y})\in E$
then flooding capacity is guaranteed to satisfy
\begin{equation}
\frac{2(k-1)}{\delta} \mc{C}_{\mc{N}_{\bs{i}}}^{m}\leq \mc{C}^m(\bs{i},\mc{N}) \leq \mc{C}_{\mc{N}_{\bs{i}}}^{m}. \label{eq:PerfBounds}
\end{equation}
\label{theorem:Thresh1}
\end{theorem}

\begin{proof} For a detailed proof see Ref.~\cite[Section IIE, Suppl. Mat]{OPGQN}. Here we sketch the basic idea: The quantity $\delta$ describes how the cardinality of any valid cut-set increases when one tries to perform a network-cut exclusively on the network-bulk. We find that for WRNs of any interest $\delta > k$, and can be computed analytically. Since $\delta$ is known analytically, it can be used to determine a minimum threshold value for single-edge capacities in the network so to ensure that a cut in the network-bulk \textit{always} generates a larger multi-edge capacity than $\mc{C}_{\mc{N}_{\bs{i}}}^m$. This then promises the performance bounds above. 
\end{proof}\\

Theorem~\ref{theorem:Thresh1} is a very useful result, and helps us to derive a global single-edge capacity constraint in order to guarantee the performance bounds in Eq.~(\ref{eq:PerfBounds}). However, if we wish to guarantee exactly optimal performance, we must place an additional constraint on user-connected edges. 

\begin{theorem}
Consider a $(k,\bs{\Lambda})$-WR quantum network. Select an end-user pair $\bs{i} = \{\bs{a},\bs{b}\}$, and demand they are sufficiently distant such that they do not share an edge. Then there exists the threshold single-edge capacity $\mc{C}_{\min}^{\prime}$ in the network bulk, and another for the user-connected edges $\mc{C}_{\min}^{\bs{i}}$ given by
\begin{gather}
\mc{C}_{\min}^{\prime}\defeq \frac{1}{\delta}\>\mc{C}_{\mc{N}_{\bs{i}}}^m, 
~~\mc{C}_{\min}^{\bs{i}}\defeq \frac{1}{\omega}\>\mc{C}_{\mc{N}_{\bs{i}}}^m,\label{eq:NeighThreshs}
\end{gather}
where we define the network parameter
\begin{equation}
\omega \defeq \frac{\delta(k-1)}{\delta - k + 1}
\end{equation}
If all single-edge capacities in the network satisfy their minimum thresholds, $\mc{C}_{\bs{xy}} \geq \mc{C}_{\min}^{\prime},\forall\>(\bs{x},\bs{y})\in E^{\prime}$ and $\mc{C}_{\bs{xy}} \geq \mc{C}_{\min}^{\bs{i}},\forall\>(\bs{x},\bs{y})\in E_{\bs{a}}\cup E_{\bs{b}}$ then flooding capacity is guaranteed to satisfy
\begin{equation}
\mc{C}^m(\bs{i},\mc{N}) = \mc{C}_{\mc{N}_{\bs{i}}}^{m}.
\end{equation}
\label{theorem:Thresh2}
\end{theorem}

\begin{proof} Once again, we point the reader towards a detailed proof in Ref.~\cite[Section IIE, Suppl. Mat]{OPGQN}. The slightly stricter condition for user-connected edges is enforced to counteract a worst-case scenario which gives rise to the lower-bound in Eq.~(\ref{eq:PerfBounds}). Nonetheless, this additional condition is not too invasive as it only applies to the subset of of $2k$ edges connected to either end-user. By imposing this condition, $\mc{C}^m(\bs{i},\mc{N}) = \mc{C}_{\mc{N}_{\bs{i}}}^m$ can be completely guaranteed.
\end{proof}

\subsection{Threshold Theorems for Network Parameters}

The theorems in the previous section provide valuable tools for understanding single-edge capacity requirements for WRNs so that performance bounds or optimal performance can be guaranteed. The threshold capacities derived are useful, but it is even more useful to identify a relationship between end-to-end network performance and \textit{physical properties} of the network nodes and channels. Hence, our goal is to translate these abstract threshold theorems into tangible relationships between physical channel parameters and end-to-end performance.

\begin{corollary}
Consider a $(k,\bs{\Lambda})$-WR quantum network $\mc{N} = (P,E)$, an end-user pair $\bs{i} = \{\bs{\alpha},\bs{\beta}\} $ and a desired min-neighbourhood capacity $\mc{C}_{\mc{N}_{\bs{i}}}^m$. Consider a single-edge channel property $\xi_{\bs{xy}}$ for which the point-to-point capacity $\mc{C}(\xi_{\bs{xy}})$ is monotonic. Then, if $\mc{C}_{\mc{N}_{\bs{i}}}^m$ is attainable, there exists a threshold parameter $\xi_{\mc{N}}^{*}$ such that
\begin{equation}
\xi_{\mc{N}}^*  \defeq \argmin_{\xi} \Big| \mc{C}(\xi) - \frac{ \mc{C}_{\mc{N}_{\bs{i}}}^m}{\delta} \Big|,
\end{equation}
which represents a maximum or minimum tolerable value of $\xi_{\bs{xy}}$ for any channel in the network:
\begin{equation}
\xi_{\mc{N}}^* \hspace{-1mm} =
\begin{cases}
 \xi_{\mc{N}}^{\max},  & \text{\emph{$\mc{C}(\xi)$ is decreasing}},\\
\xi_{\mc{N}}^{\min},  & \text{\emph{$\mc{C}(\xi)$ is increasing}}.
\end{cases}
\end{equation}
If $\xi_{\mc{N}}^*$ is obeyed for all $(\bs{x},\bs{y})\in E$ then the flooding capacity is guaranteed to satisfy
\begin{equation}
\frac{2(k-1)}{\delta}\mc{C}_{\mc{N}_{\bs{i}}}^{m} \leq \mc{C}^m(\bs{i},\mc{N}) \leq \mc{C}_{\mc{N}_{\bs{i}}}^{m}. 
\end{equation}
\label{corollary:Thresh1_Xi}
\end{corollary}

\begin{proof} Via Theorem~\ref{theorem:Thresh1}, we know that there exists a threshold capacity $\mc{C}_{\min} = \mc{C}_{\mc{N}_{\bs{i}}}^m/\delta$ which when respected throughout the network ensures that the performance bounds in Eq.~(\ref{eq:PerfBounds}) hold. 
Now, consider a physical property of a single network edge $(\bs{x},\bs{y})\in E$ denoted by $\xi \rightarrow \xi_{\bs{xy}}$, e.g.~channel length. Suppose that the single-edge capacity is a monotonic function of the single-edge property $\xi_{\bs{xy}}$. Then the threshold capacity $\mc{C}_{\min}$ can be translated into a threshold condition on $\xi_{\bs{xy}}$,
since we can write
\begin{equation}
\mc{C}_{\bs{xy}} = \mc{C}(\xi_{\bs{xy}}) \geq \mc{C}_{\min}.
\end{equation}
Therefore, there must exist a critical threshold parameter $ \xi = \xi_{\mc{N}}^*$ for which the single-edge threshold capacity is exactly satisfied,
\begin{equation}
\mc{C}(\xi_{\mc{N}}^*) = \mc{C}_{\min}. \label{eq:ThreshEq}
\end{equation}
The quantity $\xi_{\mc{N}}^*$ thus represents some limiting feature of each network edge necessary to uphold the optimal performance bounds. 

While we may not know exactly what form the single-edge capacity function takes, we can still determine the threshold value $\xi_{\mc{N}}^*$ as the value of $\xi$ which satisfies Eq.~(\ref{eq:ThreshEq}). 
Yet, we must be slightly careful here, and so let us define our capacity function more formally. The single-edge capacity $\mc{C}(\xi)$ is a function which maps single-edge network parameters $\xi$ from a domain $\mc{X}$ of possible values $\xi \in \mc{X}$, to a codomain $\mc{Y}$ of potential values $\mc{C}(\xi) \in \mc{Y}$. This codomain is necessarily a subset of the set of non-negative real numbers $\mathbb{R}_0^+$ so that the outputs of $\mc{C}$ represent meaningful capacity values. More precisely,
\begin{equation}
\mc{C} : \mc{X} \rightarrow \mc{Y} \subseteq \mathbb{R}_0^+. 
\end{equation}
We are only interested in capacity functions $\mc{C}$ which are monotonic with respect to $\xi$ so this is necessarily a one-to-one correspondence.

We may then state the following: Consider a desired min-neighbourhood capacity $\mc{C}_{\mc{N}_{\bs{i}}}^m$ and a single-edge capacity function $\mc{C}(\xi)$ which is monotonic with respect to the network parameter $\xi$. A  threshold parameter value $\xi_{\mc{N}}^{*}$ will exist if and only if $\mc{C}_{\mc{N}_{\bs{i}}}^m/\delta$ falls within the codomain $\mc{Y}$. Otherwise, there will not exist a physical value of $\xi$ for which $\mc{C}_{\mc{N}_{\bs{i}}}^m$ is attainable. 
With these considerations in mind, we can state that a network threshold parameter $\xi_{\mc{N}}^*$ can be appropriately computed such that
\begin{align}
\xi_{\mc{N}}^*  = \argmin_{\xi} \Big| \mc{C}(\xi) - \frac{ \mc{C}_{\mc{N}_{\bs{i}}}^m}{\delta} \Big| \iff \frac{\mc{C}_{\mc{N}_{\bs{i}}}^{m}}{\delta} \in \mc{Y}.
\end{align}

If $\mc{C}(\xi_{\bs{xy}})$ is a monotonically decreasing function, then $\xi_{\mc{N}}^*$ must be a \textit{maximum threshold value} $\xi_{\mc{N}}^{\max}$ because increasing it further would decrease the capacity below the threshold, i.e.~$\mc{C}(\xi_{\mc{N}}^{*} + \varepsilon) < \mc{C}_{\min}$, where $\epsilon > 0$. If $\mc{C}(\xi_{\bs{xy}})$ is a monotonically increasing function, the opposite is true and $\xi_{\mc{N}}^*$ must be a \textit{minimum threshold value} $\xi_{\mc{N}}^{\min}$ because decreasing it further would reduce the capacity below the threshold, i.e.~$\mc{C}(\xi_{\mc{N}}^{*} - \varepsilon) < \mc{C}_{\min}$ where $\epsilon > 0$. This completes the result.
\end{proof}\\

Analogously, we can translate Theorem~\ref{theorem:Thresh2} with respect to a tangible, threshold network parameter.

\begin{corollary}
Consider a $(k,\bs{\Lambda})$-WR quantum network $\mc{N} = (P,E)$, an end-user pair $\bs{i} = \{\bs{\alpha},\bs{\beta}\} $ and a desired min-neighbourhood capacity $\mc{C}_{\mc{N}_{\bs{i}}}^m$. Consider a single-edge channel property $\xi_{\bs{xy}}$ for which the point-to-point capacity $\mc{C}(\xi_{\bs{xy}})$ is monotonic. Then, if $\mc{C}_{\mc{N}_{\bs{i}}}^m$ is attainable, there exist threshold parameters $\xi_{\mc{N}}^{*}$ and $\xi_{\mc{N}_{\bs{i}}}^{*}$ which represent maximum or minimum tolerable values of $\xi_{\bs{xy}}$ for edge in the network-bulk or user-connected edge respectively:
\begin{equation}
(\xi_{\mc{N}}^*, \xi_{\mc{N}_{\bs{i}}}^*) \hspace{-1mm} \defeq 
\begin{cases}
 (\xi_{\mc{N}}^{\max}, \xi_{\mc{N}_{\bs{i}}}^{\max})  & \text{\emph{$\mc{C}(\xi_{\bs{xy}})$ is decreasing}},\\
(\xi_{\mc{N}}^{\min}, \xi_{\mc{N}_{\bs{i}}}^{\min})  & \text{\emph{$\mc{C}(\xi_{\bs{xy}})$ is increasing}}.
\end{cases}
\end{equation}
If $\xi_{\mc{N}}^*$ and $\xi_{\mc{N}_{\bs{i}}}^*$ are obeyed in their respective sub-networks, then the flooding capacity is guaranteed to satisfy
\begin{equation}
\mc{C}^m(\bs{i},\mc{N}) = \mc{C}_{\mc{N}_{\bs{i}}}^{m}. 
\end{equation}
\label{corollary:Thresh2_Xi}
\end{corollary}

\begin{proof} This corollary is simply a translation of Theorem~\ref{theorem:Thresh2} with respect to a single-edge network property $\xi_{\bs{xy}}$ for which the capacity function is monotonic. It follows an identical logic to the previous corollary, except now we must identify two different threshold values by solving the equations,
\begin{align}
&\mc{C}(\xi_{\mc{N}}^*) = \mc{C}_{\min}^{\prime} = \frac{\mc{C}_{\mc{N}_{\bs{i}}}^m}{\delta},\\ 
&\mc{C}(\xi_{\mc{N}_{\bs{i}}}^*) = \mc{C}_{\min}^{\bs{i}} = \frac{\mc{C}_{\mc{N}_{\bs{i}}}^m}{\omega}.
\end{align}
The maximum or minimum threshold behaviour and argument minimisation follow identically as before. 
\end{proof}\\

Corollaries~\ref{corollary:Thresh1_Xi} and \ref{corollary:Thresh2_Xi} therefore identify maximum or minimum threshold parameters, which when respected throughout the network are able to guarantee tight performance bounds or optimal performance. An interesting threshold parameter might be the maximum link-length, the maximum thermal noise that can be tolerated at a receiver, etc. These theorems reveal extremely useful relationships between the end-to-end rate and physical properties of interest; providing invaluable guidance for future quantum network design. 

\subsection{Threshold Theorems with Capacity Bounds}

Up to this point we have been deriving exact threshold quantities via the assumption that an expressions for the point-to-point capacity $\mc{C}(\xi)$ with respect to some physical property $\xi$ is known. Yet, when the exact nature of a quantum channel capacity is not known it is necessary to make use of capacity bounds, as have been explored in this paper. Indeed, consider a single-edge channel property $\xi$ for which the point-to-point capacity $\mc{C}(\xi)$ is monotonically bounded by lower and upper bounding functions $F \in \{ F_l, F_u\}$ respectively. We introduce the following ``cost function" to evaluate the difference between a desired min-neighbourhood capacity $\mc{C}_{\mc{N}_{\bs{i}}}^m$ and a capacity bound,
\begin{equation}
\mc{B}_{F_j}(\xi, x) \defeq \left| F_j({\xi}) - \frac{\mc{C}_{\mc{N}_{\bs{i}}}^m}{x} \right|.
\end{equation}
Here $x$ is a free parameter used to scale the capacity bound in accordance with the network connectivity properties. The goal of constructing this cost function is to ensure that $\mc{B}_{F_j}(\xi, x)$ is minimised (tends to zero) when $\xi \rightarrow \xi_{\mc{N}}^{*}$. Clearly, we may define this function in many different ways provided that this behaviour is maintained. In our numerical studies, we choose to minimise the log-ratio instead,
\begin{equation}
\mc{B}_{F_j}(\xi,x) =| \log( x F_{j}(\xi)) - \log(\mc{C}_{\mc{N}_{\bs{i}}}^{m})|.
 \end{equation}

Given a cost function, the goal is to determine bounds on the threshold value $\xi_{\mc{N}}^*$ which helps to guarantee optimal performance bounds. Hence, we define the quantity,
\begin{equation}
{{\xi}_{F_{j}}^{*}(x)} \defeq \argmin_{\xi} {\mc{B}_{F_{j}}(\xi, x)}
\end{equation}
where $j \in \{ l, u \}$, which presents a lower or upper an upper or lower bound on $\xi_{\mc{N}}^*$. 

Thanks to single-edge capacity bounds, we can generalise Corollaries~\ref{corollary:Thresh1_Xi} and \ref{corollary:Thresh2_Xi} to derive conditions for end-to-end performance on networks consisting of quantum channels whose exact capacities are not known. This results in the following corollary.

\begin{corollary}
Consider Corollary~\ref{corollary:Thresh1_Xi} and a single-edge channel property $\xi_{\bs{xy}}$ for which the point-to-point capacity $\mc{C}(\xi_{\bs{xy}})$ is monotonically bounded by lower and upper bounding functions $F \in \{ F_l, F_u\}$ respectively. The threshold parameter $\xi_{\mc{N}}^{*}$ satisfies
\begin{equation}
{\xi_{F_{j}}^{*}(\delta)} \leq \xi_{\mc{N}}^{*} \leq  {\xi_{F_{k}}^{*}(\delta)} 
\end{equation}
where $j\neq k \in \{l, u\}$ such that 
\begin{equation}
\begin{cases}
{\xi_{F_u}^{\max}(\delta) \leq \xi_{\mc{N}}^{\max} \leq \xi_{F_l}^{\max}(\delta)},  & \text{\emph{$\mc{C}(\xi)$ is decreasing}},\\
{\xi_{F_l}^{\min}(\delta) \leq \xi_{\mc{N}}^{\min} \leq \xi_{F_u}^{\min}(\delta)},  & \text{\emph{$\mc{C}(\xi)$ is increasing}}.
\end{cases}
\label{eq:BoundOrder}
\end{equation}
\label{corollary:Thresh1_Bound}
\end{corollary}

\begin{proof} 
We may not know $\mc{C}(\xi_{\bs{xy}})$ exactly, but instead have a pair of single-edge bounding functions $F_u$ and $F_l$ which similarly depend on the same single-edge parameter 
\begin{equation}
F_l(\xi_{\bs{xy}}) \leq \mc{C}(\xi_{\bs{xy}}) \leq F_u(\xi_{\bs{xy}}).
\end{equation}
When this is the case, it's not possible to determine the network threshold parameter $\xi_{\mc{N}}^*$ exactly. Instead, we can provide upper and lower bounds on its value. We can always demand that
\begin{gather}
F_u(\xi_{\bs{xy}}) \geq \mc{C}(\xi_{\bs{xy}}) \geq F_l(\xi_{\bs{xy}}) \geq  \frac{\mc{C}_{\mc{N}_{\bs{i}}}^m}{\delta},
\end{gather}
be satisfied for all $ (\bs{x},\bs{y})\in E$. Given that $F_u$ and $F_l$ are both monotonic with respect to $\xi_{\bs{xy}}$, by independently solving the following pair of equations
\begin{align}
F_l(\xi) =  \frac{\mc{C}_{\mc{N}_{\bs{i}}}^m}{\delta}, \text{ and } F_u(\xi) = \frac{\mc{C}_{\mc{N}_{\bs{i}}}^m}{\delta}. \label{eq:Conds_UL}
\end{align}
there will exist a pair of unique values $\xi_{F_u}^{*}$ and $\xi_{F_l}^{*}$ which appropriately bound the threshold parameter.

Once again it is important to be careful here, and so let us define our capacity bounding functions more formally. The function $F_k$ (for $k\in\{u,l\}$) is that which maps single-edge network parameters $\xi$ from a domain $\mc{X}$ of possible values $\xi \in \mc{X}$, to a codomain $\mc{Y}_{F_k}$ of potential bounding values $F_k(\xi) \in \mc{Y}_{F_k}$. The meaningful part of this codomain is necessarily a subset of the set of non-negative real numbers $\mathbb{R}_0^+$ so that the outputs of $F_k$ represent valid capacity values (bounding functions may have less well behaved codomains, but this is always rectifiable by truncating such regions). More precisely,
\begin{equation}
F_{k} : \mc{X} \rightarrow \mc{Y}_{{F_k}}\subseteq \mathbb{R}_0^+. 
\end{equation}
Since we are only interested in capacity bounding functions $F_{k}(\xi)$ which are monotonic with respect to $\xi$, then this is necessarily a one-to-one correspondence.
Consider a desired min-neighbourhood capacity $\mc{C}_{\mc{N}_{\bs{i}}}^m$ and an appropriate single-edge capacity bounding function $F_k(\xi)$ which is monotonic with respect to the network parameter $\xi$. A network threshold parameter bounding value $\xi_{F_k}^{*}(\delta)$ for $k\in \{u,l\}$ will exist if and only if $\mc{C}_{\mc{N}_{\bs{i}}}^m/\delta$ falls within in the codomain $\mc{Y}_{F_k}$. Otherwise, there will not exist a physical value of $\xi$ for which $\mc{C}_{\mc{N}_{\bs{i}}}^m$ is attainable. 

Therefore, a network threshold parameter $\xi_{\mc{N}}^*$ is bounded according to $\xi_{F_u}^{*}$ and $\xi_{F_l}^{*}$ such that
\begin{align}
\xi_{F_k}^{*} \defeq \argmin_{\xi} \mc{B}_{F_k}(\xi,\delta) \iff \frac{\mc{C}_{\mc{N}_{\bs{i}}}^{m}}{\delta} \in \mc{Y}_{F_k}.
\end{align}
where  $k \in \{u,l\}$ is used to indicate whether it is derived using the upper or lower bounding function $F_k$ respectively. 

As a result, for physical single-edge properties $\xi$ for which the capacity function is monotonic, there exists a critical threshold parameter in the network $\xi_{\mc{N}}^*$ for which the capacity conditions in Theorem~\ref{theorem:Thresh1} are satisfied, and thus the end-to-end capacity is guaranteed to satisfy Eq.~(\ref{eq:PerfBounds}). Even when the capacity function is not exactly known, we can instead use upper and lower capacity bounding functions $F_u$ and $F_l$ to appropriately bound $\xi_{\mc{N}}^*$. The nature of these bounds depends on the increasing or decreasing nature of the monotonic bounding functions as shown in the corollary.
\end{proof}\\

Finally, an analogous result can be presented for exactly guaranteed optimal performance.

\begin{corollary}
Consider Corollary~\ref{corollary:Thresh2_Xi} and a single-edge channel property $\xi_{\bs{xy}}$ for which the point-to-point capacity $\mc{C}(\xi_{\bs{xy}})$ is monotonically bounded by lower and upper bounding functions $F \in \{ F_l, F_u\}$ respectively. Then the threshold parameter $\xi_{\mc{N}}^{*}$ satisfies
\begin{equation}
{\xi_{F_{j}}^{*}(\delta)} \leq \xi_{\mc{N}}^{*} \leq  {\xi_{F_{k}}^{*}(\delta)} 
\end{equation}
and the user-connected edge threshold parameter satisfies
\begin{equation}
{\xi_{F_{j}}^{*}(\omega)} \leq \xi_{\mc{N}_{\bs{i}}}^{*} \leq  {\xi_{F_{k}}^{*}(\omega)} 
\end{equation}
where $j\neq k \in \{l, u\}$ label the order of the bounds and are the same as in Eq.~(\ref{eq:BoundOrder}).
\label{corollary:Thresh1_Bound}
\end{corollary}

\begin{proof}
This result is identical to Corollary~\ref{corollary:Thresh1_Bound} with the additional bounding of the user-connected threshold quantity $\xi_{\mc{N}_{\bs{i}}}^{*}$ in order to guarantee optimal performance, which can be achieved in the same way.
\end{proof}\\

For the networks, capacity bounds and critical parameters studied in this paper, monotonicity was satisfied; allowing us to exploit these results. However, it is still possible to glean valuable information about threshold parameters when the capacity function (or its bounding functions) is not monotonic with respect to a single-edge network property $\xi_{\bs{xy}}$. 
When this is the case, Eq.~(\ref{eq:Conds_UL}) will not have unique solutions, but there will exist a number of suitable values of $\xi$ for each bounding function. This gives rise to \textit{threshold ranges} of values of $\xi$ for which the flooding capacity will have performance guarantees. This extension is intuitive, and it is left for future studies wherever it is physically relevant.

\section{Minimum Nodal Densities\label{sec:NodalDens}}

Nodal density is defined as the average number of nodes contained per unit area of the network. For a network $\mc{N}$ consistent of $n$ nodes over an area $A$, the network nodal density $\rho_{\mc{N}}$ satisfies,
\begin{equation}
\rho_{\mc{N}} \defeq \frac{n}{A}.
\end{equation} 
In the context of quantum networks, nodal density is particularly vital as it is deeply connected to the distribution of network link lengths and connectivity behaviour.

The threshold theorems described in Ref.~\cite{OPGQN} and Appendix~\ref{sec:ThreshTheorems} allow for the derivation of a critical properties for \textit{individual channels} in a quantum WRN. This is remarkably useful for quantum network design, as it motivates and guides the construction of large-scale networks through properties local to each node and each network channel. However, our threshold theorems are unable to directly access \text{global network features} such as nodal density. To directly translate Theorems~\ref{theorem:Thresh1} and \ref{theorem:Thresh2} for global network properties would require the ability to express single-edge capacities as functions of global network features. If possible, such an approach would need a statistical treatment and is beyond the scope of our research.

Fortunately, it is still possible to access important nodal density information. A crucial threshold parameter that can be considered through threshold theorems is the maximum link-length, $d_{\mc{N}}^{\max}$. Clearly, $d_{\mc{N}}^{\max}$ is intimately linked to the nodal density. If $d_{\mc{N}}^{\max}$ is large, then the network may adopt sparser configurations; if it is small, the network will need to be dense in order to maintain performance and connectivity. As a result, existence of a maximum link-length constrains the possible nodal densities that a network can adopt. 

Consider a class of network $\textsf{N} = \{ \mc{N}_j \}_j$ which obeys well defined connectivity rules (how nodes and edges are connected) and a maximum link-length, $d_{\mc{N}}^{\max}$. Then we can define a \textit{minimum nodal density} $\rho_{\mc{N}}^{\min}$ as the least dense way that any $\mc{N}_{j} \in \textsf{N}$ can be constructed according to these constraints. More precisely,
\begin{equation}
\rho_{\mc{N}}^{\min} \defeq \min_{\mc{N} \in \textsf{N}} \rho_{\mc{N}}
\end{equation}
where the minimisation is performed over all possible instances of the network class. Solving this minimisation may be called finding the sparse construction of $\textsf{N}$.

For network models with a maximum link-length, the minimum nodal density can always be expressed as
\begin{equation}
\rho_{\mc{N}}^{\min} = \xi (d_{\mc{N}}^{\max})^{-2}. \label{eq:MinND}
\end{equation}
Here, $\xi$ is some characteristic function of the network class being considered, and is found by solving the sparse construction. In general, this task is extremely challenging or the number of possible network instances simply too large to minimise. For WRNs, it is actually rather easy. Indeed, the regular and orderly structure imposed by weak-regularity ensures that the sparsest construction of its network is readily achieved by maximising the length of every edge and adopting the most symmetrical/ordered geometric orientation. 

For the WRNs considered in this work, we utilise lower-bounds on $\xi$ based on the sparse construction of infinite WRNs defined over infinite areas. Indeed, for the $k=6$ design used in Fig.~\ref{fig:ADWRN} we use $\xi \geq 2/\sqrt{3}$, while for the $k=8$ design in Fig.~\ref{fig:WR_TL} we find that $\xi\geq 2$. Using these values and Eq.~(\ref{eq:MinND}), it is possible to identify a relationship between nodal density and optimal network performance by supplementing a value of $d_{\mc{N}}^{\max}$ that guarantees specific performance bounds. For more details and architecture specific derivations, see the Supplementary Material of Ref.~\cite{OPGQN}.

\section{Compound Thermal-Loss Channels \label{sec:CompDer}}

Consider a compound channel of $N$ thermal-loss channels $\mc{E}_{\tau_j,\bar{n}_j}$ each with transmissivity $\tau_{j}$ and output thermal photon number $\bar{n}_j$ for $j\in\{1,\ldots, N\}$. A compound channel of these $N$ thermal-loss channels can be summarised into a single thermal-loss channel with total transmissivity and thermal noise parameters $\tau_{\text{tot}}$ and $\bar{n}_{{\text{tot}}}$. More precisely,
\begin{equation}
\mc{E}_{\tau_{\text{tot}}, \bar{n}_{\text{tot}}} = \mc{E}_{\tau_N,\bar{n}_N} \circ \mc{E}_{\tau_{N-1},\bar{n}_{N-1}} \circ \cdots \circ \mc{E}_{\tau_1,\bar{n}_1}. \label{eq:Comp}
\end{equation}
We wish to determine the relationship between $\tau_{\text{tot}}$, $\bar{n}_{{\text{tot}}}$ and the individual sub-channel properties. 

We can do this by inspecting the action of the compound channel on a single-mode Gaussian state, $\rho$. The relationship proves to be independent of the initial state considered, therefore this is sufficient without loss of generality. Let this initial single-mode state have the covariance matrix (with zero first moments) $V_0^{\text{in}}$
which is a real, symmetric, positive definite matrix. A thermal-loss channel $\mc{E}_{\tau_j,\bar{n}_j}$ equates to the mixing of the input state on a $\tau_j$ transmissive beam-splitter with an environmental thermal state of $\bar{n}_{j}/(1-\tau_j)$ mean photons. 
Hence, we can monitor the transformations induced by the compound channel in Eq.~(\ref{eq:Comp}) by recursively applying each transformation to the initial state $V_0^{\text{in}}$. 

After the first channel $\mc{E}_{\tau_1,\bar{n}_1}$, the output state has the covariance matrix
\begin{equation}
V_1^{\text{out}} = \tau_1 V_0^{\text{in}} + \epsilon_1  I_2,
\end{equation}
where $I_2$ is the $2\times 2$ identity matrix, and we have defined the noise quantity for the $j^{\text{th}}$ sub-channel
\begin{equation}
\epsilon_j \defeq  \Big(\frac{\bar{n}_{j}}{|1-\tau_j|}+\frac{1}{2} \Big)|1-\tau_j| =  \bar{n}_{j}+\frac{1}{2}|1-\tau_j|.
\end{equation}
All subsequent sub-channels result in the following sequence of transformations,
\begin{align}
&V_2^{\text{out}} = \tau_2 V_1^{\text{out}} + \epsilon_2 I_2,\\
&\hspace{2mm}\vdots \nonumber \\
&V_{N-1}^{\text{out}} = \tau_{N-1} V_{N-2}^{\text{out}} + \epsilon_{N-1} I_2, \\
&V_N^{\text{out}} = \tau_N V_{N-1}^{\text{out}} + \epsilon_N I_2.
\end{align}
By inspection of these recursive formula, it is clear that the final output state $V_N^{\text{out}}$ can be rewritten in the form
\begin{equation}
V_N^{\text{out}} = \Big(\prod_{j=1}^N \tau_j\Big) V_0^{\text{in}} + \xi_N \cdot I_2, \label{eq:CompForm}
\end{equation}
where we define $\xi_N \in \mathbb{R}$ is a thermal additive noise coefficient at the $N^{\text{th}}$ channel output state. This function is recursive, as it compounds the noise properties from all previous channels. For $j=1$, it takes the initial value $\xi_1 = \left(\bar{n}_{1}+\frac{1}{2}\right)|1-\tau_1|$, while for all further values $ 1 < j \leq N$ it takes the form
\begin{align}
\xi_j = \tau_j \xi_{j-1} + \bar{n}_{j}+\frac{1}{2}|1-\tau_j|.
\end{align}
This can be easily solved via recursive techniques, and the final additive noise term after $N$ thermal-loss channels is given by
\begin{equation}
\xi_N = \epsilon_N + \sum_{j=1}^{N-1} \bigg( \epsilon_{N-j} \prod_{i=N-j+1}^N\tau_i \bigg).
\end{equation}

It is then possible to derive relationships between $\tau_{\text{tot}}$, $\bar{n}_{\text{tot}}$ and all the individual $\tau_{j}$, $\bar{n}_{j}$ values by equating Eq.~(\ref{eq:CompForm}) to the transformation induced by a single thermal-loss channel with total loss/noise properties,
\begin{equation}
V_N^{\text{out}} = \tau_{\text{tot}} V_0^{\text{in}} + \Big(\bar{n}_{\text{tot}}+\frac{1}{2}|1-\tau_{\text{tot}}| \Big) I_2, \label{eq:TotForm}
\end{equation}
Then these parameters can be quickly identified. As a result, we find that they admit the forms
\begin{align}
\tau_{\text{tot}} &\defeq \prod_{j=1}^N \tau_j,\\
\bar{n}_{\text{tot}} &\defeq {\xi_N} - \frac{1}{2}{|1-\tau_{\text{tot}}|}.
\end{align}
These results can then be specifically applied to the compound channels in the main text. They are clearly independent of the initial state and are thus universal for compound thermal-loss channels.

\section{Compound Amplitude Damping Channels \label{sec:DampDer}}
An $N$ length compound channel of amplitude damping channels each with damping probability $p_j$ can be simply rewritten as a single amplitude damping channel with a total damping probability
\begin{equation}
\mc{E}_{p_{\text{tot}}} = \mc{E}_{p_N} \circ \mc{E}_{p_{N-1}} \circ \cdots \circ \mc{E}_{p_1}. \label{eq:ADComp}
\end{equation}
The relationship between $p_{\text{tot}}$ and the damping probability of each sub-channel can be easily derived. Indeed, let us consider a generic single qubit state as an input into the first channel,
\begin{equation}
\rho^{\text{in}}_0 = \begin{pmatrix} 1 - \gamma & c^* \\ c & \gamma \end{pmatrix}.
\end{equation}
where $\gamma, c \in \mathbb{C}$. A single amplitude damping channel invokes the transformation,
\begin{align}
\rho^{\text{out}}_1 &= \sum_{i=0,1} K_i \rho_0^{\text{in}} K_i^{\dag},\\
 &= \begin{pmatrix} 1 - (1-p_1)\gamma & c^* \sqrt{1-p_1} \\ c \sqrt{1-p_1} & \gamma (1-p_1) \end{pmatrix},
\end{align}
where $K_i$ are the Kraus operators of the channel, such that $K_0 = \ket{0}\!\bra{1} + \sqrt{1-p}\ket{1}\!\bra{1}$ and $K_1 = \sqrt{p} \ket{0}\!\bra{1}$. It is clear we can rewrite this output state in the form,
\begin{equation}
\rho^{\text{out}}_1 = \begin{pmatrix} 1 - \gamma_1 & c_1^* \\ c_1 & \gamma_1 \end{pmatrix}, 
\end{equation}
where $\gamma_1 = {1-p_1} \gamma$ and $c_1 = c \sqrt{1-p_1}$. Hence, the subsequent action of $N-1$ further amplitude damping channels will result in the output state
\begin{equation}
\rho^{\text{out}}_N = \begin{pmatrix} 1 - \gamma_N & c_N^* \\ c_N & \gamma_N \end{pmatrix},
\end{equation}
where the parameters of the $N^{\text{th}}$ output state are
\begin{equation}
\gamma_N = \gamma \prod_{j=1}^N (1-p_j),~~c_N = c \prod_{j=1}^N \sqrt{1-p_j}.
\end{equation}
As a result, we can equate the action of the $N$ compound channel in Eq.~(\ref{eq:ADComp}) to that of a single amplitude damping channel with total damping probability,
\begin{equation}
p_{\text{tot}} \defeq 1 - \prod_{j=1}^N (1-p_{j}).
\end{equation}

\section{Practical CV-QKD Setups \label{sec:LO_App}}

CV quantum communication protocols which make use of coherent detection (such as homodyne or heterodyne measurements) require the use of a local-oscillator (or phase reference), as discussed in the main text.  
A phase reference allows the sender and receiver to exploit both quadratures of the mode, and can be established via a transmitted local-oscillator (TLO) or local local-oscillator (LLO). A TLO is an additional mode which is co-propagated along with signal-mode from the sender to receiver, carrying the relevant phase information. An LLO uses interleaving signal pulses with bright reference pulses which are used to reconstruct the local-oscillator locally at the receiver \cite{AdvCrypt}. The uniqueness of these techniques lead to unique noise/loss sources at the receiver, manifesting in different performance characteristics. For details, see Appendix B of Ref.~\cite{FS}.

\begin{table}[t!]
\begin{tabular}{ |l|c|c| } 
 \hline
\textit{\hspace{8.5mm}Parameter} & ~\textit{Symbol}~ & ~\textit{Value}~ \\ 
 \hline  \hline
Wavelength & $\lambda$ & 800~nm  \\  \hline
Detector Efficiency & $\tau^{r}$ & 0.8 \\  \hline
Detector Shot-noise & $\nu^{r}$ & $\begin{array}{l} 1\hspace{1.75mm}\text{- Homodyne}\\ 2\hspace{1.75mm}\text{- Heterodyne}\end{array}$ \\  \hline
Detector bandwidth & $W$ & 100 MHz\\ \hline
Channel noise & $\bar{n}_B$ & 0.002 \\ \hline
Noise Equivalent Power & NEP & 6 pW/$\sqrt{\text{Hz}}$\\ \hline
LO Power & $P_{\text{LO}}$ & 100 mW\\ \hline
Line width & $l_{\text{W}}$ & 1.6 KHz\\ \hline
Clock & $C$ & 5 MHz \\ \hline
Pulse Duration  & $\Delta t$, $\Delta t_{\text{LO}}$ & $10$ ns \\   \hline
Modulation  & $\mu$ & $10$  \\   \hline
\end{tabular}
\caption{Parameter table for realistic CV-QKD protocols through optical-fibre channels. }
\label{table:Setups}
\end{table}

Let us consider CV quantum communication between two arbitrary nodes $\bs{x},\bs{y} \in P$ using the physical channel direction $\bs{x} \rightarrow\bs{y}$. Firstly, we may assume that upon transmission, both LO techniques induce negligible loss or noise, i.e.~$\bar{n}_{\bs{x}}^{s,\text{TLO}} \approx \bar{n}_{\bs{x}}^{s,\text{LLO}} \approx 0$ and $\tau_{\bs{x}}^{s,\text{TLO}} \approx \tau_{\bs{x}}^{s,\text{LLO}} \approx 0$. But at the receiver, there will be decoherence associated with sub-optimal detection and potential phase errors. The detector will not be perfectly efficient, limited by lossy fibre-couplings and limited quantum efficiency. This is captured by the detection loss, given by $\tau_{\bs{y}}^{r} = \tau_{\text{eff}}$. 

There will also be noise induced by the detector and the local-oscillator. The use of an LLO introduces phases errors when it is being reconstructed at the receiver, since this reconstruction will never be perfect. As such, we conceive a phase noise parameter
\begin{equation}
\Theta_{\text{ph}} \defeq \frac{\pi (\mu - 1) l_{\text{W}}}{C},
\end{equation}
where $\mu$ describes the modulation of the transmitted pulses, $C$ is the operational rate (clock) and $l_{\text{W}}$ is the average linewidth of the light source. 

Meanwhile, both techniques will be exposed to electronic noise due to sub-optimal detection and imperfect fibre-couplings. Let us define the electronic noise parameter,
\begin{equation}
 \Theta_{\text{el}} \defeq \frac{\nu_{\text{det}} \text{NEP}^2 W\Delta t_{\text{LO}}}{2h\nu P_{\text{LO}}^{\text{det}}}.
\end{equation}
Here, $\nu_{\text{det}}$ is the detector shot-noise where $\nu_{\text{det}} = 1$ for homodyne and $\nu_{\text{det}} = 2$ for heterodyne. This quantity also depends on the noise equivalent power ($\text{NEP}$), the bandwidth $W$, the duration of the local-oscillator pulse $\Delta t_{\text{LO}}$, the frequency of the light $\nu$, and the power of the local-oscillator at detection, $P_{\text{LO}}^{\text{det}}$. Since the LLO is reconstructed locally, the power at detection is simply the desired power $P_{\text{LO}}^{\text{det}} = P_{\text{LO}}$. On the other hand, using a TLO diminishes its power at detection due to loss suffered throughout its transmission $P_{\text{LO}}^{\text{det}} = \eta_{\bs{xy}} \tau_{\text{eff}} P_{\text{LO}}$. 

Collecting these noise sources for each method, we can write
\begin{align}
&\bar{n}_{\bs{y}}^{r,\text{LLO}} \approx \eta_{\bs{xy}} \tau_{\text{eff}} \Theta_{\text{ph}} + \Theta_{\text{el}},~~\bar{n}_{\bs{y}}^{r,\text{TLO}} \approx \frac{\Theta_{\text{el}}}{\eta_{\bs{xy}} \tau_{\text{eff}}}.
\end{align}

It is clear that there is a trade-off between the phase-errors induced by a LLO and the electronic noise induced by a TLO. The reciprocal dependence of $\bar{n}_{\bs{y}}^{r,\text{TLO}}$ on the channel transmissivity means that as longer distances it will introduce greater levels of noise (which is precisely the behaviour shown in Fig.~\ref{fig:WR_TL}). As such, it becomes wise to make use of the LLO method, yet its technical requirements are somewhat more demanding. Table~\ref{table:Setups} collects typical values for the setup parameters which contribute to these noise quantities when considering optical-fibre connections.

\end{document}